\pdfoutput=1
\documentclass[aps,pra,reprint,groupaddress,superscriptaddress,notitlepage,onecolumn,nofootinbib]{revtex4-1}

\usepackage{algpseudocode}

\usepackage{enumerate}

\usepackage[utf8]{inputenc}
\usepackage{amssymb,dsfont}
\usepackage{graphicx,subcaption}
\usepackage{amsmath}
\usepackage{amsbsy}

\usepackage{amsthm}
\newtheorem{theorem}{Theorem}

\newtheorem{proposition}{Proposition}
\newtheorem{corollary}{Corollary}

\theoremstyle{definition}
\newtheorem{definition}{Definition}
\newtheorem{example}{Example}

\usepackage{bbm}
\usepackage{bm}

\DeclareMathOperator{\poly}{poly}

\usepackage[hyperindex,breaklinks,colorlinks]{hyperref}

\usepackage{newfloat}

\DeclareFloatingEnvironment[
    fileext=loa,
    listname=List of Algorithms,
    name=ALG,
    placement=tbhp,
]{algorithm}

\usepackage{tikz}
\usetikzlibrary{arrows,shapes,decorations,automata,backgrounds,petri,calc,trees,positioning,fit,intersections}
\usetikzlibrary{decorations.pathmorphing}
\usetikzlibrary{decorations.markings}

\usepackage{color}

\def\ket#1{\left| #1 \right\rangle}

\DeclareMathOperator{\enc}{Enc}
\DeclareMathOperator{\EffEnc}{EffEnc}
\def\union{\textsf{Union}}
\def\contains{\textsf{Contains}}
\def\append{\textsf{Append}}
\def\convert{\textsf{Convert}}
\def\effcontains{\textsf{EffContains}}
\def\calculate{\textsf{Calculate}}
\def\setgen{\textsf{SetGen}}
\def\Eppstein{\textsc{Eppstein}}
\def\reduce{\textsc{Reduce}}
\def\checkcycles{\textsc{Check}}
\def\alg{\textsc{Alg}}
\def\qalg{\textsc{QAlg}}
\def\hybridalg{\textsc{HybridAlg}}
\def\trivial{\textsc{Trivial}}
\def\rectree{\textsc{RecTree}}
\def\trivred{\textsc{TrivRed}}
\def\edgeselect{\textsc{EdgeSelect}}

\colorlet{algbgcolour}{black!10}
\tikzset{blackdot/.style={circle,inner sep=1.5pt,fill}}

\begin{document}
\title{A hybrid algorithm framework for small quantum computers with application to finding Hamiltonian cycles}
\author{Yimin Ge}
\email{yimin.ge@mpq.mpg.de}
\affiliation{Max-Planck-Institut f\"{u}r Quantenoptik, Hans-Kopfermann-Str. 1,
85748 Garching, Germany}
\author{Vedran Dunjko}
\email{v.dunjko@liacs.leidenuniv.nl}
\affiliation{LIACS, Leiden University, Niels Bohrweg 1, 2333 CA Leiden, Netherlands}
\begin{abstract}
Recent works \cite{DGC2018} have shown that quantum computers can polynomially speed up certain SAT-solving algorithms even when the number of available qubits is significantly smaller than the number of variables.   
Here we generalise this approach. We present a framework for hybrid quantum-classical algorithms which utilise  quantum computers significantly smaller than the problem size. Given an arbitrarily small ratio of the quantum computer to the instance size, we achieve polynomial speedups for classical divide-and-conquer algorithms, provided that certain criteria on the time- and space-efficiency are met. We demonstrate how this approach can be used to enhance Eppstein's algorithm for the cubic Hamiltonian cycle problem, and achieve a polynomial speedup for any ratio of the number of qubits to the size of the graph. 
\end{abstract}

\maketitle

\section{Introduction}

Although fully scalable quantum computers may be far off, small quantum computers may be achievable in the forseeable future. 
Such devices may able to provide solutions to toy or specialised problems of small size (e.g. in quantum chemistry \cite{QChemIntro}),
it was however until recently unclear whether they could also be utilised for speeding up more general and common  computations of much larger problem
instances. 
Indeed, quantum and classical algorithms usually exploit global structures inherent to the problem, and it is generally difficult to utilise small quantum computers without breaking that structure. 
For example, the ability to factor $n/10$-digit integers is unlikely to be of much help for the task of factoring $n$-digit integers. 
One would therefore naively expect that for structured problems, small quantum computers would only be useful for small problem sizes. 

Recently in \cite{DGC2018}, it was show that this is not generally true: given a quantum computer with only $M$ qubits, it was shown that one can obtain a significant speedup of Sch\"oning's algorithm for solving 3SAT involving $n\gg M$ variables. 
More precisely, the speedup can be expressed in terms of the ratio $c=M/n$ of available qubits to the problem size, and it was shown that an asymptotic polynomial speedup can be achieved for arbitrarily small values of $c$. 
The latter is non-trivial, since it was also shown that   
employing a naive  approach to speed up the classical algorithm breaks the exploited problem structure, thus resulting in no improvement unless $c$ is quite large. 

One of the main insights of \cite{DGC2018} was that classical divide-and-conquer algorithms inherently don't suffer from this ``threshold effect'' since they naturally maintain the structure of the problem despite breaking it into smaller subproblems, and are thus well-suited for being enhanced using small quantum computers. Yet, to achieve genuine speedups, the quantum subroutines employed must meet stringent criteria for space- and time-efficiency, which in general are non-trivial to fulfil and, in the case of \cite{DGC2018}, required specialised data-structures and careful memory management. Moreover, \cite{DGC2018} exclusively considers the example of Sch\"oning's 3SAT algorithm, and while it demonstrated how divide-and-conquer structures can in principle be exploited, it left open a formal characterisation of the criterion for when this approach works, and whether other examples beyond Sch\"oning's 3SAT algorithm exist where similar speedups can be obtained\footnote{Note that since reductions of one NP-complete problem to another generally incur significant polynomial slowdowns, and we only expect polynomial speedups, a speedup for 3SAT does not imply a speedup for other NP-complete problems. Consequently, each problem and algorithm must be treated individually.}. Given the ubiquity of classical divide-and-conquer algorithms, a general framework for developing such hybrid algorithms to obtain speedups using only small quantum computers would thus be highly desirable.

In this work, we formalise and generalise the criteria for the hybrid approach of \cite{DGC2018}, and show that indeed, there are other problems and algorithms that can be enhanced with that approach, using only small quantum computers.  
Specifically, we develop a general framework for constructing hybrid algorithms that speed up certain classical divide-and-conquer algorithms using significantly fewer qubits than the problem size, which also makes precise the relation between the speedup obtained on the one hand, and the space requirements and runtime of the underlying quantum subroutines on the other.  
We then apply this formalism to finding Hamiltonian cycles on cubic graphs, which is another fundamental NP-complete problem. This provides the first example of the applicablity of these techniques beyond the example of Sch\"oning's 3SAT algorithm given in \cite{DGC2018}. 
Our framework operates on the algorithmic level and is distinct from the circuit-level techniques of \cite{2016_Bravyi,arXiv:1904.00102}, which aim to simulate general quantum circuits on fewer qubits. The efficiency of the latter depends on sparseness or decomposability assumptions on the original circuits which will in general prevent speedups for the algorithms we consider here. 

We provide our formalism for constructing hybrid algorithms from classical algorithms  in the form of a ``toolkit'' comprising two parts. The first part, which we term the \emph{divide-and-conquer hybrid approach}, shows general criteria for the kind of classical algorithms one can speed up using our techniques, and relates the speedup to the number of  available qubits.  
The second part, which we term \emph{efficient reversible set generation}, comprises specialised data-structures. These are specifically designed to bridge the gap between two seemingly irreconciliable properties required of the quantum algorithm for a polynomial speedup in the divide-and-conquer hybrid approach: reversibility on the one hand, and being extremely space-efficient on the other. The set generation procedure we develop makes the task of developing quantum algorithms suitable for our hybrid approach significantly easier:  
we show that it suffices to find space-efficient implementations for a few problem-specific quantum operations. 

Finally, we illustrate how this framework is applied to speed up Eppstein's algorithm \cite{Eppstein} for the cubic Hamiltonian cycle problem with only a small quantum computer. The cubic Hamiltonian cycle problem asks if a given cubic graph of $n$ vertices has a \emph{Hamiltonian cycle}, i.e. a cycle  visiting every vertex exactly once. This problem is NP-complete, and  a special case of the general Hamiltonian cycle problem (where no restrictions on the maximum degree of the graph is assumed), which in turn is closely linked to the travelling salesman problem. Brute-force search requires $O(n!\poly(n))$ time,   
there is however also a trivial path-search algorithm with runtime $O(2^n\poly(n))$. In 2004, Eppstein \cite{Eppstein} gave a divide-and-conquer algorithm of runtime $O(2^{n/3}\poly(n)) = O(1.2599^n\poly(n))$, which heavily exploits the cubic structure of the graph. Quantum speedups for Eppstein's algorithm have previously been obtained using arbitrarily-sized quantum computers \cite{PhysRevA.95.032323}. In this work, we obtain a polynomial speedup using only $M=c n$ qubits for arbitrarily small $c>0$. 
 
The outline of the remainder of this paper is as follows. In Section~\ref{sec:summary}, we give a brief overview of the results and clarify some notation. In Section~\ref{sec:sha}, we formulate the divide-and-conquer hybrid approach for a general class of classical algorithms. In Section~\ref{sec:iterative}, we provide the details of the efficient and reversible  set generation procedure. In Section~\ref{sec:eppstein}, we apply these tools to Eppstein's algorithm for the cubic Hamiltonian cycle problem. Finally, we close the paper with some concluding remarks and open questions in Section~\ref{sec:conclusion}.

\section{Overview of results}\label{sec:summary}

We briefly summarise the results and main ideas of this paper. 

Sections~\ref{sec:sha} and \ref{sec:iterative} set up the general framework for designing hybrid algorithms for quantum computers significantly smaller than the problem size. Specifically, Section~\ref{sec:sha} introduces the divide-and-conquer hybrid approach, and outlines the general criteria for the kind of classical algorithms which our framework is applicable to.  Theorem~\ref{thm:sha} then formalises the trade-off between the number of available qubits, the space-requirement of the underlying quantum subroutine, and the speedup obtained. The main idea of this hybrid approach is to take a classical divide-and-conquer algorithm, which calls itself on ever smaller problem instances, and to replace the recursive call with a suitable quantum algorithm once the problem instance is sufficiently small to fit the number of available qubits. While the basic idea is simple,  in order to obtain a polynomial speedup over the original classical algorithm, the replacing quantum algorithm has to fulfill strict criteria for space-efficiency whilst also being polynomially faster. In many cases, the main contribution to the quantum algorithm's space-requirement  comes from the necessity of generating and storing large sets.

Section~\ref{sec:iterative} then shows how to do the latter efficiently. In particular, Theorem~\ref{thm:setgen} provides a (classical) reversible set generation routine which can be used to obtain quantum algorithms that are compatible with a polynomial speedup when used in Theorem~\ref{thm:sha}. To that end, specialised set encodings are first introduced which are designed to overcome the main challenge of such an implementation: the ability to uncompute encodings of previously generated sets without resulting in either large computational overheads or large memory requirements. 

Section~\ref{sec:eppstein} then provides an example of how these tools can be applied in practice. Eppstein's algorithm for finding Hamiltonian cycles on cubic graphs is a classical divide-and-conquer algorithm that naturally fits the framework of Theorem~\ref{thm:sha}. The main ingredient for a speedup then becomes a polynomially faster quantum algorithm that solves this problem using sufficiently few qubits. Theorem~\ref{thm:quantumAlg} proves the existence of such a quantum algorithm. The proof of Theorem~\ref{thm:quantumAlg} heavily utilises the set generation routine of Theorem~\ref{thm:setgen}, which reduces the task to implementing a small number of problem-specific, i.e. graph-theoretic, operations. Theorem~\ref{thm:quantumAlg}, together with Theorem~\ref{thm:sha}, then immediately imply a polynomial speedup of Eppstein's algorithm using only significantly fewer qubits than the size of the graph, which is formally stated in Theorem~\ref{thm:eppsteinSpeedup}.

Throughout this paper, we will use standard bra/ket notation for quantum states. We will also use bra/ket notation in the context of classical reversible circuits, since they can be seen as special cases of quantum circuits. 

Moreover, for simplicity of notation, we will often make several notational simplifications.
First, we often simply write $\ket{0}$ for $\ket{0}^{\otimes L}$ for any known $L\in\mathbb N$. The value of $L$ will always be clear from context. 
Second, for operators acting on some registers of a multi-register state, we will normally not explitly write the complementing identiy operators (e.g., we will simply write $A\ket a\ket b \ket c\ket d$ instead of $(\mathds 1\otimes A\otimes\mathds 1)\ket a\ket b\ket c\ket d$ if $A$ acts on the middle two registers). It will always be clear from context which registers which operators act on.

\section{The divide-and-conquer hybrid approach}\label{sec:sha}

In this section, we formalise the \emph{divide-and-conquer hybrid approach}, generalising the techniques of \cite{DGC2018}, for designing hybrid algorithms using only quantum computers significantly smaller than the problem size. The main idea is to take a classical divide-and-conquer\footnote{Many algorithms that are not a priori given in this form can be formulated as such.} algorithm that calls itself on ever smaller (effective) problem sizes, and replace the recursive calls with a quantum algorithm once the problem size becomes sufficiently small.

Let $\mathcal P$ be a countable set, $\mathcal A: \mathcal P\rightarrow\{0,1\}$ be a decision problem\footnote{For simplicity, we formulate the divide-and-conquer hybrid approach for decision problems here. The approach can be generalised to algorithms with more general outputs, subject to size constraints of the output.}, and $n:\mathcal P\rightarrow \mathbb N$  be a problem parameter.  
We refer to $n(P)$ as the \emph{problem size} of $ P$.   

\begin{algorithm}[htbp]
\colorbox{algbgcolour}{\parbox{\linewidth}{\begin{algorithmic}[1]
\Procedure{$\alg$}{$P$}
  \State if $\trivial(P)=1$
  \State \indent return $f(P)$
  \State else
  \State\label{algAlinerecursive} \indent return $g(\alg(R_1(P)),\ldots,\alg(R_l(P)))$
\EndProcedure
\end{algorithmic}}}
\caption{General algorithm for the divide-and-conquer hybrid approach}\label{alg:algsha}
\end{algorithm}

Suppose that $\mathcal A(P)$ can be decided by a classical recursive algorithm\footnote{In general, the recursive algorithm can also take additional parameters, but these can be incorporated into $P$.} of the form given in Alg.~\ref{alg:algsha}, 
where $l\geq 2$ is an integer, $R_1,\ldots,R_l:\mathcal P\rightarrow \mathcal P$, $g:\{0,1\}^l\rightarrow\{0,1\}$, and $ f,\trivial: \mathcal P\rightarrow\{0,1\}$. We  assume that  $\trivial(P), R_1(P),\ldots,R_l(P)$ can be calculated in $O(\poly (n(P)))$ time, and that $f(P)$ can be calculated in $O(\poly(n(P)))$ time if $\trivial(P)=1$.  The maps $R_1,\ldots,R_l$ can be thought of as ``reduction operations'', mapping the problem to a smaller instance, whereas $\trivial(P)$ signifies if $P$ is sufficiently simple to be solved directly. 
We assume that for all $P\in\mathcal P$ and $i=1,\ldots,l$,  $n(R_i(P))\leq n(P)$. 

The runtime of such divide-and-conquer algorithms can often be bounded by introducing an \emph{effective problem size} $s: \mathcal P \rightarrow \mathbb N$. 
In general, $s(P)$ and $n(P)$ can be different, we assume however that for all $P\in\mathcal P$, 
$s(P)\leq n(P)$. We assume that both $n(P)$ and $ s(P)$ can be calculated in  time $O(\poly(n(P)))$. 
  We moreover assume that  
there is a universal constant $C\in\mathbb N$ such that $s(P)\leq C$ implies $\trivial(P) = 1$.

	To ensure that $\alg(P)$ has a runtime of the form $O(2^{\gamma s(P)}\poly(n(P)))$ for some constant $\gamma>0$, we assume that there exist integers $k>0$, $C_{ij}>0$ for $1\leq i \leq l, 1\leq j\leq k$, and $l_1,\ldots,l_k\in\{1,\ldots,l\}$  such that for all $P\in\mathcal P$ with $\trivial(P)=0$, there exists some $j\in\{1,\ldots, k\}$ such that for all $i=1,\ldots, l$, we have 
\begin{equation}\label{eq:sReduction}
	s(R_i(P)) \leq s(P) - C_{ij} \quad\text{or}\quad \trivial(R_i(P))=1
\end{equation} 
for $i=1,\ldots,l_j$,  
and $\trivial(R_i(P))=1$ for $i=l_j+1,\ldots,l$. Here, $j\in\{1,\ldots,k\}$ labels one of $k$ possible cases for effective problem size reductions, and $l_j\leq l$ the effective number of recursive branches of that case, whereas the positive integers $C_{ij}$ are lower bounds on the decrease of the effective problem size, guaranteeing that the algorithm terminates.

Under these assumptions, it is easy to derive the stated upper bound on the runtime  of running $\alg(P)$. Indeed, \eqref{eq:sReduction} implies a recursive runtime bound $T(s(P))$  depending only on $s(P)$ given by
\begin{equation}
	T(s') \leq \max_{j=1,\ldots,k} ( T(s'-C_{1j}) + \cdots + T(s'-C_{l_jj})) + O(\poly(n(P)))
\end{equation}
and $T(s')=O(\poly n(P))$ for $s'\leq C$. Using standard methods for solving recurrence relations \cite{Bentley:1980:GMS:1008861.1008865}, this leads to a runtime of $T(s') = O(2^{\gamma s'}\poly(n(P)))$ for some $\gamma>0$. Thus, $\alg(P)$ has a runtime of $O(2^{\gamma s(P)}\poly(n(P)))$.

	Our aim is to provide criteria for when a reduction of the value of $\gamma$, i.e. a polynomial speedup, can be achieved.

\begin{theorem}[Divide-and-conquer hybrid approach] \label{thm:sha}
	Suppose that there is a quantum algorithm $\qalg$ that decides $\mathcal A(P)$  using at most $G(s(P),n(P))$ qubits in time $O(2^{\gamma_Q s(P)}\poly(n(P)))$ for some constant $\gamma_Q \in [0, \gamma)$. Suppose that $G:[0,\infty)\times[0,\infty)\rightarrow[0,\infty)$ has the property that
	for all nonnegative integers $s\leq n'\leq n$, $G(s,n')\leq G(s,n)$.  
	 Suppose moreover that there exists some $\tilde\lambda\in(0,1)$ such that for all $\lambda\in[0,\tilde \lambda]$ and $n\in\mathbb N$, 
	\begin{equation}
		G(\lambda n, n) = nF(\lambda) + O(\log n)
\end{equation}
for some strictly monotonically increasing and continuously differentiable $F:[0,\tilde\lambda]\rightarrow [0,\infty)$ such that $F(0)=0$ and $F'$ is bounded away from $0$.  

 	 Let $c\in(0,F(\tilde\lambda))$ be an arbitrary constant. Then, given a quantum computer with $M=cn(P)$ qubits, there exists a hybrid quantum-classical algorithm that decides $\mathcal A(P)$  in a runtime of $O (\max(2^{\gamma s(P) - f(c)n(P)}, 2^{\gamma_Q s(P)})\poly(n(P)))$, where $f(c) = (\gamma-\gamma_Q)F^{-1}(c)>0$. 
 	 In particular, $\mathcal A(P)$ can be decided in $O (2^{(\gamma  - f(c))n(P)}\poly(n(P)))$. 
\end{theorem}
\begin{proof}
The proof is based on the ideas developed in \cite{DGC2018} in the context of 3SAT. The main idea of the hybrid algorithm is to call $\qalg(P')$ instead of $\alg(P')$ in the recursive step of $\alg$ when $s(P')$ is sufficiently small. 

	By assumption, running $\qalg(P')$ for any $P'\in\mathcal P$ with $n(P')\leq n(P)$ requires at most $G(s(P'),n(P'))\leq G(s(P'),n(P)) \leq n(P)F(s(P')/n(P))+a\ln n(P)$ qubits for some constant $a>0$. 
Hence, if $n(P')\leq n(P)$, then
\begin{equation}
	s(P')\leq n(P)F^{-1}\left(c-\frac{a\ln n(P)}{n(P)}\right) =: \tilde s 
\end{equation}
is a sufficient condition for being able to run $\qalg(P')$ with $M=cn(P)$ qubits. 
Note that since $F$ is strictly increasing and continuously differentiable with its derivative bounded away from $0$, the same applies to $F^{-1}$. 
 Thus, by the mean value theorem, $\tilde s = F^{-1}(c) n(P) -O(\log n(P))$. 

Let $\hybridalg(P')$ be the algorithm which calls $\qalg(P')$ if $s(P')\leq \tilde s$ and $\alg'(P')$ otherwise, where $\alg'$ is the same algorithm as $\alg$ except that in line~\ref{algAlinerecursive} of $\alg$, the calls to $\alg$ are replaced by calls to $\hybridalg$. Note that since $n(R_i(P'))\leq n(P')$ for all $P'\in\mathcal P$ and $i=1,\ldots,l$, $M=cn(P)$ qubits suffice to run $\hybridalg(P)$. 

We now analyse the runtime of $\hybridalg$. Note that its runtime can be bounded by the function $T_H(s(P))$ depending only on $s(P)$, where $T_H$ is given recursively by 
\begin{equation}
	T_H(s') =O(2^{\gamma_Q s'}\poly(n(P)))
\end{equation}
for $s'\leq \tilde s$ and
\begin{equation}
	T_H(s') \leq \max_{j=1,\ldots,k} ( T_H(s'-C_{1j}) + \cdots + T_H(s'-C_{l_jj})) + O(\poly(n(P)))
\end{equation}
for $s'>\tilde s$. Using standard recurrence relation techniques, we therefore obtain
\begin{equation}
	T_H(s') = O(2^{\gamma(s'-\tilde s) + \gamma_Q \tilde s}\poly(n(P))) = O(2^{\gamma s' - (\gamma-\gamma_Q){\tilde s}}\poly(n(P)))  = O(2^{\gamma s' - f(c)n(P)}\poly(n(P))),
\end{equation}
 for $s'>\tilde s$, which proves the claim. 
\end{proof} 

We remark that the assumptions of Theorem~\ref{thm:sha} can be relaxed in various ways. First, it is sufficient for $\qalg$ to only decide $\mathcal A(P)$ for $P\in\bigcup_{i=1}^l R_i(\mathcal P)$. Second, $F$ being continuously differentiable can be relaxed to $F^{-1}$ being locally Lipschitz-continuous. 

Note moreover that the assumption that $F'$ is bounded away from $0$ ensures that the degree of the polynomial overhead of the hybrid algorithm can be bounded by a constant independent of $c$. 
More precisely, if  $F'(\lambda)\geq \kappa >0$ for all $\lambda\in[0,\tilde\lambda]$, then the hybrid algorithm decides $\mathcal A(P)$ in a runtime of $O(\max(2^{\gamma s(P) - f(c)n(P)},2^{\gamma_Q s(P)}) n(P)^{O(1/\kappa)})$. 

\begin{example}\label{ex:slogns}
	Suppose that in Theorem~\ref{thm:sha}, running $\qalg(P)$ requires $O(s\log(n/s) + s+\log n)$ qubits, where here we just write $s,n$ instead of $s(P),n(P)$ for simplicity. In that case, $G(s,n) = As\ln(n/s) + Bs+O(\log n)$ for some constants $A,B>0$. 
	 Then, $F(\lambda) = A\lambda\ln(1/\lambda) + B\lambda$, which is monotonically increasing on $(0,e^{B/A-1})$. 
	It can be shown that 
$
	F^{-1}(c) = -c/(AW_{-1}(-ce^{-B/A}/A))
$,
where $W_{-1}$ is the $-1$ branch of the Lambert $W$ function. It is easy to see that for small values of $c$, $F^{-1}(c) =\Theta(c/\log(1/ c))$. 

The curious expression $O(s\log(n/s)+s)$ stems from the information-theoretic cost of encoding a subset of $\{1,\ldots, O(n)\}$ of size $O(s)$. It is moreover the scaling obtained in Ref~\cite{DGC2018} (see Example~\ref{ex:PBS} below) and also later in Section~\ref{sec:iterative} and \ref{sec:eppstein}. 
\hfill$\blacksquare$
\end{example}

\begin{example}\label{ex:slogn}
	Suppse that $\qalg$ requires $\geq s\log_2 n$ qubits instead. Note that this does not satisfy the requirements of Theorem~\ref{thm:sha}. Then, with the same notation as in the proof of Theorem~\ref{thm:sha}, $\tilde s = cn/\log_2 n$, and hence 
	\begin{equation}\label{eq:subpolyspeedup}
		T_H(n) = O(2^{(\gamma- \frac c{\log_2 n})n}\poly(n))).
	\end{equation}
	Note that this does not yield a polynomial speedup over $\alg$, since the value of $\gamma$ is not reduced by a constant. 
	
	The importance of this example lies in that while a qubit scaling of $O(s\log n)$ is, in many cases, easy to achieve (e.g. through storing an ordered list of $O(s)$ numbers in $\{1,\ldots,O(\poly n)\}$), it however does not lead to a polynomial (albeit still asymptotic) speedup. A similar result can also be seen for a scaling of $O(s\log s)$.  
	\hfill$\blacksquare$
\end{example}

Note that the strictness of the space requirement for $\qalg$ to obtain a polynomial speedup   comes  from the premise of only having a quantum computer of size $M=cn$. Note that if in Example~\ref{ex:slogn}, we were given a quantum computer with $M=cn\log n$ qubits instead, a polynomial speedup would still be obtained. The strength of the speedup therefore critically depends what is considered a ``natural'' scaling of $M$ relative to $n$.
In many of the typical applications, the search space of typical classical algorithms (e.g. brute-force search) can be enumerated using $O(n)$ classical bits. Using amplitude amplification,  one can then often obtain a quantum algorithm using $O(n)$ qubits that is usually polynomially (and often quadratically) faster than the corresponding classical algorithm \cite{2004_Ambainis}. In these cases, $M=cn$ is the natural scaling because if the scaling were to be relaxed to $M=cq(n)$ qubits with $q(n)$ being superlinear in $n$, then  even for arbitrarily small $c$, the above  quantum algorithm would require asymptotically fewer qubits than the hybrid algorithm, which would be inconsistent with the notion that $M$ should be significantly smaller than the number of qubits required by a full quantum algorithm. 

\begin{example}\label{ex:PBS}
	We now show that the results of \cite{DGC2018} also fit in this paradigm. We use the nomenclature of \cite{DGC2018}, and refer the interested reader to \cite{DGC2018,2011_Moser} for further details. 

	In \cite{DGC2018}, $\alg$ was taken to be the  algorithm  from \cite{2011_Moser} for the Promise-Ball-SAT problem. There, $P=(F,\mathbf x, r)$ comprises an $n$-variable $3$SAT formula $F$, a trial assignment $\mathbf x\in\{0,1\}^n$ and a radius $r\in\{1,\ldots,n\}$. The effective problem size $s(P)$ was simply taken to be its radius $r$ of $P$. 
	
	In \cite{2011_Moser}, it was shown that there are two possible cases ($k=2$) in the recursive algorithm. In the first case, two subproblems of radius $r-1$ are created. In the second case, at most $ t^2 2^{\Delta}$ subproblems of radius $r-\Delta$ are created, where $\Delta\in \mathbb N$ is a constant chosen below and $t=3\Delta$. 
	
	Then, with the notation of this section, $l= t^2 2^{\Delta}$, $k=2$, $l_1=2$, $l_2=l$, and 
	\begin{equation}
		C = \left( 
			\begin{array}{cc}
				1 & \Delta \\
				1 & \Delta \\
				0 & \Delta \\
				\vdots & \vdots \\
				0  & \Delta
			\end{array}		
		\right). 
	\end{equation}
	
	This leads to a recursive runtime bound of 
	\begin{equation}
		T(r) \leq \max\left( 2T(r-1), t^22^\Delta T(r-\Delta)\right),
	\end{equation}
	leading to a runtime of $T(r) = O((2t^{2/\Delta})^r\poly(n))=O(2^{(1+\varepsilon)r}\poly(n))$, where $\varepsilon = 2\Delta^{-1}\log_2 (3\Delta)$. Note that $\varepsilon\rightarrow 0$ as $\Delta\rightarrow \infty$.

	Ref~\cite{DGC2018} then constructs a quantum algorithm  solving the Promise-Ball-SAT problem in time $O(2^{\gamma_Q r}\poly(n))$ with $\gamma_Q=\log_2(3)/2 < 1$, and using at most $O(r\log (n/r) +r +\log n)$ qubits. Thus, Theorem~\ref{thm:sha} and the observation in Example~\ref{ex:slogns} implies that Promise-Ball-SAT can be solved in time $O(\max(2^{ (1+\varepsilon) r - \tilde f(c) n}, 2^{\gamma_Q r})\poly(n))$  with $\tilde f(c) = \Theta(c/\log(1/c))$. Note that $\Delta$ can be chosen such that $\varepsilon < \tilde f(c)/2$. The results from \cite{2002_Dantsin}, which reduce 3SAT to Promise-Ball-SAT, then imply that given a quantum computer with $M=cn$ qubits, $n$-variable 3SAT can be solved in time 
$O(2^{(\gamma-f(c))n}\poly(n))$,	 
	where $f(c)=\tilde f(c)/2 $ and $\gamma=\log_2(4/3)$. This is a polynomial speedup of the 3SAT algorithm obtained in \cite{2011_Moser} (which in turn can be seen as a derandomised version of Sch\"oning's algorithm \cite{1999_Schoning}), which is the central result of \cite{DGC2018}.\hfill$\blacksquare$
\end{example}

\section{Efficient and reversible set generation}\label{sec:iterative}

The previous section highlights the importance of the space-efficiency of the quantum algorithm used in Theorem~\ref{thm:sha}. 
In many instances, this quantum algorithm require the storing and manipulation large sets, e.g. to keep track of changes to $P$. As observed in \cite{DGC2018}, this is in general a non-trivial task when constrained by limited memory. 

In this section, we formulate a general process to space- and time-efficiently generate an encoding of a set in a reversible manner, designed to be compatible with the use of Theorem~\ref{thm:sha} to obtain a polynomial speedup using small quantum computers. 

Note that most of the required subroutines can trivially be implemented space-efficiently if one assumes that ancillary memory registers can be erased at will. However, we naturally require  our computations to be reversible. This is in general an issue, since turning non-reversible computations into reversible ones in the canonical fashion either introduces many ancillas, or can incur exponential overheads (see e.g. \cite{BTV2001}), both of which are non-starters for our needs. 
Combining the seemingly competing requirements of reversibility, small memory requirements and computational efficiency for our purposes is non-trivial. 
To illustrate the problem (see also \cite{DGC2018}), 
note that if sets $S=\{x_1,\ldots,x_r\}\subset\{1,\ldots,N\}$ are simply stored as ordered lists $\ket{x_1}\ldots\ket{x_r}$, the memory requirement of $O(r\log N)$ qubits is too large for a polynomial speedup if $r=O(s(P))$ and $N=O(n(P))$ (see Example~\ref{ex:slogn}). 
On the other hand,
 if sets are encoded as genuine sets (i.e., without storing any ordering of the elements in the set), the operation $\ket{\{x_1,\ldots, x_{i-1}\}}\ket{x_i}\mapsto \ket{\{x_1,\ldots, x_{i}\}}$ is non-reversible, because the information on which element was added last is lost. 
 The naive way to make this reversible would  be to first implement the operation $\ket{\{x_1,\ldots, x_{i-1}\}}\ket{x_i}\mapsto \ket{\{x_1,\ldots, x_{i-1}\}}\ket{x_i}\ket{\{x_1,\ldots, x_{i}\}}$ and then to uncompute the $\ket{\{x_1,\ldots, x_{i-1}\}}\ket{x_i}$ registers by applying the inverse of the circuit up to that point. It is easy to see however that this incurs a computational overhead of $O(2^i)$, which is too large. Indeed, suppose that $\setgen_i\ket 0 = \ket{\{x_1,\ldots,x_i\}}$ and $\calculate_i\ket{0}=\ket{x_i}$. Then, the naive (recursive) implementation of $\setgen_{i}$ would be to first apply $\calculate_{i}\setgen_{i-1}$ to generate $\ket{\{x_1,\ldots, x_{i-1}\}}\ket{x_i}$, then implementing $\ket{\{x_1,\ldots, x_{i-1}\}}\ket{x_i}\mapsto \ket{\{x_1,\ldots, x_{i-1}\}}\ket{x_i}\ket{\{x_1,\ldots, x_{i}\}}$, and finally applying $(\calculate_{i}\setgen_{i-1})^{-1}$ to uncompute the  ancillas. The resulting recursive runtime would be $|\setgen_{i}| > 2|\setgen_{i-1}|$, leading to an exponential gate count of $O(2^i)$. 

In this section, we describe a general formalism to overcome this problem based on the ideas of \cite{DGC2018}. Specifically, we will show how the above task can in fact be implemented with $O(r\log(N/r)+r+\log N)$ memory and $O(\poly(N))$ runtime.  Of course, one of the key aspects of our implementation is the continuous  uncomputation of any ancillas we introduce along the way once they are no longer needed. Once they are uncomputed (i.e., reset to a known initial state, say $\ket 0$), they can be re-used for later computational steps. This allows us to keep the overall number of ancillas used low. The primary challenge is to do this in a way which avoids the exponential overhead mentioned above. 

In Section~\ref{subsec:setenc}, we first introduce the data-structures which allow this suitable trade-off between space-efficiency and computational overhead of uncomputation. Specifically, we develop a space-efficient encoding of large sets which adds just enough ordering information to allow for efficient uncomputation. Afterwards, in Section~\ref{subsec:iterative}, we describe the general algorithm for generating such an encoding of a set, given only access to operations which generate single elements of the set. 

All algorithms considered in the remainder of this section are classical and will be written as reversible circuits. 
For convenience, we also introduce the following notion of reversible implementation and note the subsequent trivial observation. 

\begin{definition}
	Let $q,l,g\in\mathbb N$, $\mathcal X \subset \{0,1\}^q$ and $f:\mathcal X \rightarrow \{0,1\}^q$ be injective. We say that $f$ can be \emph{implemented reversibly using $l$ ancillas and $g$ gates} if there exists a classical reversible  circuit of at most $g$ elementary gates which for all $x\in\mathcal X$ implements the operation
	\begin{equation}
		\ket{x}\ket{0}^{\otimes l} \mapsto \ket{f(x)}\ket{0}^{\otimes l}.
	\end{equation}
\end{definition}

\begin{proposition}\label{prop:compositereversible}
	Let $q,t\in\mathbb N$, $\mathcal X_1,\ldots,\mathcal X_t \subset \{0,1\}^q$ and $f_i:\mathcal X_i \rightarrow \{0,1\}^q$, $i=1,\ldots,t$, be injective such that $f_i$ can be implemented reversibly using $l_i$ ancillas and $g_i$ gates. Suppose  that $f_i(\mathcal X_i)\subset X_{i+1}$ for $i=1,\ldots,t-1$. Then, $f_t\circ f_{t-1}\circ\cdots\circ f_1$ can be implemented reversibly using $\max_{i=1,\ldots,t} l_i$ ancillas and $g_1+\cdots+g_t$ gates. 
\end{proposition}

\subsection{Efficient set encodings}\label{subsec:setenc}

In this section we describe how to efficiently encode sets in a way which allows for efficient uncomputation whilst maintaining reversibility. We first describe ``basic'' set encodings, which use little memory but by themselves do not allow for efficient uncomputation. After that, we describe an efficient encoding composed of multiple basic encodings that allows for efficient uncomputation.

\begin{definition}
	Let $N,k\in\mathbb N$ be  positive integers with $k\leq N$, and let $S\subset \{1,\ldots,N\}$ with $|S|=k$.  Define the \emph{basic encoding} $\ket{\enc_N S}$ of $S$ to be a sequence of $\lfloor k\log_2(N/k+1)\rfloor + 2k$ trits set to
	\begin{equation}
		\ket{\enc_N S} := \ket{ (y_1)_2}\ket 2 \ket{(y_2-y_1)_2}\ket 2 \ldots \ket{(y_k-y_{k-1})_2}\ket 2\ket0\ldots\ket 0,
	\end{equation}
	where 
	$S=\{y_1,\ldots,y_k\}$  with $y_1<\cdots <y_k$, and 
	for a positive integer $y$, $\ket{(y)_2}$ denotes a sequence of $\lceil \log_2(y+1)\rceil $ trits encoding the binary representation of $y$ on the $\{0,1\}$ subspace\footnote{Note that the binary representation without leading zeros of a positive integer $y$ has $\lceil \log_2(y+1)\rceil$ bits.}. 
\end{definition}
\begin{example}
	Suppose $N=20$, $k=5$, and $S=\{6,7,10,15,17\}$. Then, $\lfloor k\log_2(N/k+1)\rfloor + 2k = 21$, and $y_1=6=110_2$, $y_2-y_1=1=1_2$, $y_3-y_2=3=11_2$, $y_4-y_3 = 5 = 101_2$, and $y_5-y_4=2=10_2$. Thus,
	\begin{equation}
		\ket{\enc_N S} = \ket{110212112101210200000}.
	\end{equation} \hfill$\blacksquare$
\end{example}

To see that $\lfloor k\log_2(N/k+1)\rfloor + 2k$ indeed suffice for $\ket{\enc_N S}$, note that the number of trits required is 
\begin{align}
	\lceil \log_2 (y_1+1)\rceil & + \lceil \log_2 (y_2-y_1+1)\rceil + \cdots + \lceil \log_2 (y_k-y_{k-1}+1)\rceil + k \\ 
	&\leq  \log_2 (y_1+1) + \log_2(y_2-y_1+1) + \cdots + \log_2(y_k-y_{k-1}+1) + 2k\\
	&\leq  k\log_2 ((y_k+k)/k) +2k  \label{eq:jensen}
	\\ &\leq k\log_2 (N/k+1) + 2k,
\end{align}
where \eqref{eq:jensen} follows from Jensen's inequality. Note that this is significantly less than the naive encoding of $S$ as an ordered list, which uses $O(k\log N)$ bits.

In \cite{DGC2018}, it was shown how to perform basic set operations on $\ket{\enc_N S}$. 
\begin{definition}
	Let $N$ be a positive integer. 
	\begin{enumerate}[(i)]
		\item For any positive integer $k\leq N$, let $\contains_{N,k}$ be the operation that performs
	\begin{equation}\label{eq:containsdef}
		\contains_{N,k}\ket{\enc_N S}\ket{x} \ket 0 = \ket{\enc_N S}\ket x \ket{x\in S?} 
	\end{equation}
	for any $S\subset\{1,\ldots,N\}$ with $|S|=k$, where the last bit on the right-hand side of \eqref{eq:containsdef} is $1$ if $x\in S$ and $0$ otherwise. 
		\item Let $\convert_{N}$ be the operation that performs
		\begin{equation}
			\convert_N\ket x\ket 0 = \ket x\ket{\enc_N\{x\}}
		\end{equation}
		for any $x\in\{1,\ldots,N\}$. 
		\item For any positive integers $k_1,k_2$ with $k_1+k_2\leq N$, let $\union_{N, k_1,k_2}$ be the operation that performs
		\begin{equation}
			\union_{N, k_1,k_2}\ket{\enc_N S_1}\ket{\enc_N S_2}\ket 0 = \ket{\enc_N S_1}\ket{\enc_N S_2}\ket{\enc_N S_1\cup S_2}
		\end{equation}
		for any disjoint $S_1,S_2\subset\{1,\ldots,N\}$ with $|S_1|=k_1$ and $|S_2|=k_2$. 
	\end{enumerate}	
\end{definition}
\begin{proposition}[\cite{DGC2018}, Lemma 2,5,6 in Supplemental Material]\label{prop:setops}
	Let $N$ be a  positive integer. Then,
	\begin{enumerate}[(i)]
		\item\label{prop:setopscontains} for any positive integer $k\leq N$, $\contains_{N,k}$ can be implemented reversibly using $O(\log N)$ ancillas and $O(\poly(N))$ gates. 
		\item $\convert_N$ can be implemented reversibly using $O(\log N)$ ancillas and $O(\poly(N))$ gates\footnote{Note that in the Supplemental Material of \cite{DGC2018}, $\convert_N$ was called $\append_0$.}.
		\item\label{prop:setopsunion} for any positive integers $k_1,k_2$ with $K=k_1+k_2\leq N$, $\union_{N, k_1,k_2}$ can be implemented reversibly using $O(K\log(N/K)+K+\log N)$ ancillas and $O(\poly(N))$ gates. 
	\end{enumerate}
\end{proposition}

Note in particular that in Proposition~\ref{prop:setops}(\ref{prop:setopscontains}) and (\ref{prop:setopsunion}), the runtime bound (i.e., the degree of the polynomial) does not depend on $k$ or $k_1,k_2$, respectively. In fact, none of the operations depend on the set-sizes in any relevant way.  

Although $\ket{\enc_N S}$ is itself space-efficient, it does not allow for a set generation procedure that is simultaneously space- and time-efficient as well as reversible, for the reasons explained at the beginning of this section. We now define a memory-structure that allows for this task. 

\begin{definition}\label{def:EffEnc}
	Let $N,k$ be  positive integers with $k\leq N$. Let $S\subset\{1,\ldots,N\}$ with $|S|=k$, and let $Z=(x_1,\ldots,x_k)$ be a permutation of the elements of $S$. Then, the \emph{efficient encoding} $\ket{\EffEnc_N Z}$ of $Z$ is definded as follows: suppose that $k$ has binary representation $k=2^{a_1}+\cdots+2^{a_s}$ with integers $\lfloor\log_2 k\rfloor = a_1 > a_2 >\cdots > a_s\geq 0$. For $j=1,\ldots,s$, let $k_j:=2^{a_1}+\cdots+2^{a_j}$. Then, $\ket{\EffEnc_N Z}$ is defined as
	\begin{equation}
		\ket{\EffEnc_N Z} := \ket{\enc_N \{x_1,\ldots,x_{k_1}\}}\ket{\enc_N \{x_{k_1+1},\ldots,x_{k_2}\}}\ldots \ket{\enc_N \{x_{k_{s-1}+1},\ldots,x_{k}\}}. 
	\end{equation}
\end{definition}

\begin{example}
 Suppose $k=13=8+4+1$ and $Z=(x_1,\ldots,x_{13}) $. Then,
 \begin{equation}
 	\ket{\EffEnc_N Z} = \ket{\enc_N \{x_1,x_2,x_3,x_4,x_5,x_6,x_7,x_8\}}\ket{\enc_N \{x_9,x_{10},x_{11},x_{12}\}}\ket{\enc_N \{x_{13}\}}.
 \end{equation}
 \hfill$\blacksquare$
\end{example}

\begin{proposition}
	For any positive integers $k\leq N$ and distinct integers $x_1,\ldots,x_k\in\{1,\ldots,N\}$, $\ket{\EffEnc_N Z}$ comprises at most $\lfloor 2k\log_2(N/k+1) \rfloor + 8k = O(k\log(N/k)+k)$ trits, where $Z=(x_1,\ldots,x_k)$.
\end{proposition}
\begin{proof}
	For any $S'\subset\{1,\ldots, N\}$, $\ket{\enc_N S'}$ comprises at most $|S'|\log_2(N/|S'| + 1) + 2|S'|$ trits. Thus, the number of trits in $\ket{\EffEnc_N Z}$ is at most
	\begin{align}
		\sum_{l=0}^{\lfloor \log_2 k\rfloor}\left( 2^l\log_2\frac {N+2^l}{2^l}  + 2^{l+1}\right)
		&= \sum_{l=0}^{\lfloor \log_2 k\rfloor} 2^l\log_2 \frac{N+2^l}{k} + \sum_{l=0}^{\lfloor \log_2 k\rfloor} 2^{l+1} + \sum_{l=0}^{\lfloor \log_2 k\rfloor} 2^l\log_2 \frac{k}{2^l}\\
		&\leq \sum_{l=0}^{\lfloor \log_2 k\rfloor} 2^l\log_2 \frac{N+k}{k}  + \sum_{l=0}^{\lfloor \log_2 k\rfloor} 2^{l+1} + k\sum_{l=0}^{\lfloor \log_2 k\rfloor} \frac{2^l}{k} \log_2\frac k{2^l} \\
		&\leq 2k\log_2\left(\frac Nk+1\right) + 4k + 2k\sum_{l=0}^{\lfloor \log_2 k\rfloor} \frac{1}{2^{\lceil \log_2 k\rceil - l}} (\lceil \log_2 k\rceil -l)\\
		&\leq  2k\log_2\left(\frac Nk+1\right) + 4k + 2k\sum_{j=0}^\infty \frac j{2^j} \\
		&\leq 2k\log_2\left(\frac Nk+1\right) + 8k,
	\end{align}
	which proves the claim. 
\end{proof}

Note that since we always assume that we work with sets of known sizes, basic set operations can be ``lifted'' from $\ket{\enc_N S}$ to $\ket{\EffEnc_N Z}$. 

\begin{definition}
	For any positive integers $N,k$ with $k\leq N$, let $\effcontains_{N,k}$ be the operation that performs 
	\begin{equation}\label{eq:effencdef}
		\effcontains_{N,k}\ket{\EffEnc Z}\ket x \ket 0 = \ket{\EffEnc Z} \ket x \ket{ x\in Z?}
	\end{equation}
	for any integer $x\in\{1,\ldots, N\}$ and distinct integers $x_1,\ldots,x_k\in \{1,\ldots N\}$, where $Z=(x_1,\ldots,x_k)$ and the last bit on the right-hand side of \eqref{eq:effencdef} is $1$ if $x=x_j$ for some $j\in\{1,\ldots,k\}$ and $0$ otherwise. 
\end{definition}
\begin{proposition}\label{prop:effcontains}
	For all positive integers $N,k$ with $k\leq N$, $\effcontains_{N,k}$ can be implemented reversibly using $O(\log N)$ ancillas and $O(\poly(N))$ gates. 
\end{proposition}
In particular, the runtime bound in Proposition~\ref{prop:effcontains} does not depend on $k$. 
\begin{proof}[Proof of Proposition~\ref{prop:effcontains}]
	This follows immediately from Proposition~\ref{prop:setops}(\ref{prop:setopscontains}): Introduce $s\leq \log_2 k = O(\log N)$ ancilla bits (where $s$ is defined as in Definition~\ref{def:EffEnc}), run $\contains_{N,2^{a_j}}$ on $\ket{\enc_N \{x_{k_{j-1}+1},\ldots,x_{k_j}\}}\ket x$ for all $j=1,\ldots,s$ (where $a_1,\ldots,a_s$ and $k_1,\ldots,k_s$ are defined as in Definition~\ref{def:EffEnc}) such that the outcome is stored in the $j$\textsuperscript{th} ancilla, apply a logical OR over the $s$ ancilla bits, and finally uncompute the $s$ ancilla bits by running $\contains_{N,2^{a_j}}^{-1}$ for $j=1,\ldots,N$. 
\end{proof}

Note that if $N'>N$ are integers and $S\subset\{1,\ldots,N\}$, then $\ket{\enc_N S}$ and $\ket{\enc_{N'}S} $ only differ in the number of additional $\ket{0}$'s at the end of the encoding. For the remainder of the paper, whenever the value of $N$ is clear from context, we will drop the subindex $N$ for simplicity of notation and simply write $\ket{\enc S}$, $\ket{\EffEnc Z}$, $\contains_k$, $\convert$, $\union_{k_1,k_2}$ and $\effcontains_k$  instead. 

\subsection{Efficient and reversible set generation} \label{subsec:iterative}

Suppose we want to generate a set $X( \nu)=\{x_1\ldots,x_r\}\subset \{1,\ldots, N\}$ of size $r$ from some input register $\ket{\nu}$. We assume that we can generate the elements iteratively, i.e. we have access to circuits that generate $x_i$ from $\nu$ and $\{x_1,\ldots, x_{i-1}\}$. As discussed in at the beginning of this section, simply using the basic encoding $\ket{\enc X}$ and adding one elment $x_i$ to $\ket{\enc X}$ at a time is problematic, as this incurs problems with reversibility or exponential computational overheads. In this subsection, we show how this can be circumvented using the $\ket{\EffEnc Z}$ encoding from the previous subsection. 

The intuition for why $\ket{\EffEnc Z}$, unlike $\ket{\enc X}$, allows for time-efficient uncomputation is that by splitting $X$ into smaller subsets, we can generate and uncompute these subsets more efficiently, thus avoiding the necessity of uncomputing the entire set for each newly added element, which can be seen as the reason for the exponential overhead in the naive reversible implementation given at the beginning of this section.

\begin{theorem}\label{thm:setgen}
	Let $N,r$ be known positive integers with $r\leq N$, and let $\mathcal I$ be a finite set. Let $X:\mathcal I \rightarrow \mathcal P_r([N])$, where $\mathcal P_r([N]) := \{ S\subset \{1,\ldots,N\}, |S|=r\}$. 
	Suppose that  for  $i=1,\ldots r$, $\calculate_i$ are reversible circuits such that for all $\nu\in \mathcal I$, there is a permutation $(x_1,\ldots , x_r)$ of the elements of $X(\nu)$ such that for all $i=1,\ldots,r$,
\begin{equation}\label{eq:calculateithmdef}
	\calculate_i \ket{\nu} \ket{\EffEnc Z_{i-1}}\ket 0\ket 0= \ket{\nu}\ket{\EffEnc Z_{i-1}}\ket{x_i}\ket 0, 
\end{equation}
where $Z_i=(x_1,\ldots, x_i)$, and the last register in \eqref{eq:calculateithmdef} comprises at most $A$ ancillas. Then, the operation
\begin{equation} \label{eq:targetmap}
	\ket{ \nu}\ket 0 \mapsto \ket{ \nu}\ket{\EffEnc Z_r}
\end{equation}
can be implemented reversibly using $O(r\log(N/r) + r + \log N+A)$ ancillas, $O(r^2)$ calls to $\calculate_i$ for some $i\in\{1,\ldots,r\}$, and $O(\poly(N))$ additional gates. 
\end{theorem}
\begin{proof}
	For any positive integer $i$, write $g(i)$ to be the largest integer $g$ such that $2^g$ divides $i$. 
	Then, for all $i\in\{1\ldots,r-1\}$ and $l\in\{0,\ldots, g(i)\}$ such that $i+2^l\leq r$, 
	let $R_{i,l}$  be the operation that performs
	\begin{equation}
		R_{i,l}\ket \nu\ket{\EffEnc Z_{i}}\ket 0 =\ket \nu\ket{\EffEnc Z_{i}}\ket{\enc \{x_{i+1},\ldots x_{i+2^l}\}}
	\end{equation}
	for all $\nu\in \mathcal I$, where $x_1,\ldots, x_r$ and $Z_1,\ldots, Z_r$ are as in the statement of the theorem. 
	Note that if $l\leq g(i)-1$, then
	\begin{equation}\label{eq:RilEffEnc}
		\ket{\EffEnc Z_{i}}\ket{\enc \{x_{i+1},\ldots, x_{i+2^l}\}} = \ket{\EffEnc Z_{i+2^l}}.
	\end{equation}	 
	We also define $R_{0,l}$ for all $l\leq \lfloor \log_2 r\rfloor$ to be the operation that performs 
	\begin{equation}
		R_{0,l}\ket \nu\ket 0 =\ket \nu\ket{\enc \{x_{1},\ldots, x_{2^l}\}} = \ket \nu \ket{\EffEnc Z_{2^l}}
	\end{equation}	
	for all $\nu\in\mathcal I$. 
	Note that  since $\convert\ket{x_i}\ket 0 = \ket{x_i}\ket{\enc \{x_i\}}$, 
	a reversible implementation of 
	 $R_{i,l}$ can be obtained recursively via
	\begin{equation}\label{eq:Ri0}
		R_{i,0} = \calculate_{i+1}^{-1}\convert\ \calculate_{i+1}
	\end{equation} for all $i\in\{0,\ldots,r-1\}$, 
	and 
	\begin{equation}\label{eq:Rilrecursion}
		R_{i,l+1} =  R_{i,l}^{-1}R_{i+2^{l},l}^{-1} \union_{2^{l},2^{l}} R_{i+2^{l},l} R_{i,l}
	\end{equation}
	for all $i\in\{0,\ldots,r-1\}$ and $l \in\{0,\ldots, g(i)-1\}$ such that $i+2^{l+1}\leq r$.	
	 Indeed, for $i\in\{1,\ldots,r\}$ and $l\leq g(i)-1$, \eqref{eq:RilEffEnc} implies that 
	\begin{equation}
		R_{i,l}\ket \nu\ket{\EffEnc Z_{i}}\ket 0  =
		\ket \nu \ket{\EffEnc Z_{i+2^l}}.
	\end{equation}		 
	  Hence, $R_{i+2^l,l}R_{i,l}$ maps $\ket \nu\ket{\EffEnc Z_{i}}\ket0\ket0\ket0$ to 
	\begin{equation}
		\ket \nu\ket{\EffEnc Z_{i}}\ket{\enc \{x_{i+1},\ldots,x_{i+2^l}\}}\ket{ \enc\{x_{i+2^l+1},\ldots,x_{i+2^{l+1}} \}}\ket0,
	\end{equation}
	and $\union_{2^l,2^l}$ maps the latter to
	\begin{equation}
		\ket \nu\ket{\EffEnc Z_{i}}\ket{\enc \{x_{i+1},\ldots,x_{i+2^l}\}}\ket{ \enc\{x_{i+2^l+1},\ldots,x_{i+2^{l+1}} \}}\ket{\enc \{x_{i+1},\ldots,x_{i+2^{l+1}} \}}.
	\end{equation}
	Finally, $ R_{i,l}^{-1}R_{i+2^l,l}^{-1}$ uncomputes the registers containing $\ket{\enc \{x_{i+1},\ldots,x_{i+2^l}\}}\ket{ \enc\{x_{i+2^l+1},\ldots,x_{i+2^{l+1}} \}}$, proving Eq.~\eqref{eq:Rilrecursion} for $i\in\{1,\ldots,r\}$. A similar argument also shows that \eqref{eq:Rilrecursion} holds for $i=0$.  
	
	Suppose that $r$ has  binary expansion $r=2^{a_1}+2^{a_2}+\cdots+2^{a_s}$, with integers $\lfloor \log_2 r\rfloor = a_1 > a_2 >\cdots > a_s = g(r)$. For $j=1,\ldots,s$, define $r_j:=2^{a_1} + \cdots + 2^{a_j}$, and 
	\begin{equation}\label{eq:defR}
		R_j := R_{r_{j-1},a_j}\cdots R_{r_1,a_2}R_{0, a_1}.
	\end{equation}
	Note that $a_j\leq g(r_{j-1})-1$ for all $j\in\{2,\ldots,s\}$, so \eqref{eq:RilEffEnc} implies that \eqref{eq:defR} is well-defined. Note moreover that $R_j\ket \nu\ket 0 = \ket \nu \ket{\EffEnc Z_{r_j}}$ for all $j\in\{1,\ldots,s\}$. In particular, $R_s$ performs the desired operation \eqref{eq:targetmap}.
	
	 It follows from \eqref{eq:Rilrecursion} that each individual $R_{i,l}$ constitutes a reversible implementation using at most $O(r\log(N/r)+r+\log N+A)$ ancillas.  Hence, Proposition~\ref{prop:compositereversible} implies that $R_s$ can be implemented reversibly using at most $O(r\log(N/r)+r+\log N+A)$ ancillas. 
	
	To bound the number of calls to $\calculate_{i'}$ and additional gate count of $R_s$, let $L_{i,l}$ be the number of calls of $R_{i,l}$ to $\calculate_{i'}$ for some $i'\in\{1,\ldots,r\}$, and let $M_{i,l}$ be the number of additional gates of $R_{i,l}$, respectively. Let $L_l = \max \{L_{i,l} : i\in\{0,\ldots,r-1\}, l\leq g(i), i+2^l\leq r\}$ and $M_l = \max\{M_{i,l} : i\in\{0,\ldots,r-1\}, l\leq g(i), i+2^l\leq r\}$. Eq.~\eqref{eq:Ri0} and \eqref{eq:Rilrecursion} clearly imply $L_l=2\cdot 4^l$. Moreover, since $\union_{2^l,2^l} $ can be implemented using at most $p(N)$ gates, where $p$ is a polynomial independent of $l$, it follows from \eqref{eq:Rilrecursion} that $M_{l+1}\leq 4M_l + p(N)$, implying $M_l=O(4^l\poly(N))$. Thus, $R_s$ uses at most 
	\begin{equation}
		L_{a_1}+\cdots+L_{a_s} \leq 2(1+4+\cdots+4^{\lfloor \log_2 r\rfloor}  )= O(r^2)
	\end{equation}
	calls to $\calculate_{i'}$ for some $i'\in\{1,\ldots,r\}$, and 
	\begin{equation}
		M_{a_1}+\cdots+M_{a_s} = O((1+4+\cdots+4^{\lfloor \log_2 r\rfloor})\poly(N)) = O(\poly(N))
	\end{equation}
	additional gates. 
\end{proof}

Theorem~\ref{thm:setgen} generates the efficient encoding $\ket{\EffEnc Z_r}$ of the set $X(\nu)$ instead of the simple encoding $\ket {\enc X(\nu)}$. For most applications, the former is sufficient, since the value of $r$ is known and simple set queries for checking properties of $X(\nu)$ (e.g. checking if $X(\nu)$ contains certain elements) are generally just as simple to implement reversibly using $\ket{\EffEnc Z_r}$ as with $\ket{\enc X(\nu)}$ (see Proposition~\ref{prop:effcontains}). We remark however that $\ket{\EffEnc Z_r}$ can be converted to $\ket{\enc X(\nu)}$ using a sequence of calls to $\union$ (for appropriate set sizes) and uncomputations. Since at most $O(\log r)$ union operations are required, is easy to see that this can be implemented with $O(\poly(N))$ calls to $R_{i,l}$ (as defined in the proof of Theorem~\ref{thm:setgen}) for suitable values of $i,l$, and $O(\poly(N))$ additional gates.

\begin{corollary}
	With the same notation as in Theorem~\ref{thm:setgen}, the operation
	\begin{equation}
		\ket \nu\ket 0 \mapsto \ket \nu\ket{\enc X(\nu)}
	\end{equation}
	can be implemented reversibly using $O(r\log (N/r)+r+\log N+A)$ ancillas, $O(\poly(N))$ calls to $\calculate_i$, and $O(\poly(N))$ additional gates. 
\end{corollary}

\section{Speedup of Eppstein's algorithm}\label{sec:eppstein}

In this section, we provide an example of how to use the toolkit -- the divide-and-conquer hybrid approach from Section~\ref{sec:sha} and the set-generation procedure from Section~\ref{sec:iterative} -- to polynomially speed up Eppstein's algorithm \cite{Eppstein} for the cubic Hamiltonian cycle problem using a small quantum computer. The problem asks whether a given cubic graph $G=(V,E)$ has a Hamiltonian cycle, i.e. a cycle going through every vertex exactly once.

\subsection{Eppstein's algorithm}\label{subsec:classical_eppstein}

In this section, we review Eppstein's classical algorithm for solving this problem in time $O(2^{n(G)/3}\poly(n(G)))$, where $n(G)$ denotes the number of vertices  of a graph $G $. 

Note first of all that, without loss of generality, one can assume that the graph is triangle-free, since triangles can be removed by merging the three vertices of a triangle into a single vertex. 

 Eppstein's algorithm  introduces the concept of ``forced'' edges that a Hamiltonian cycle has to contain. In other words, if an edge is forced, we are only looking for Hamiltonian cycles which contain that edge.

\begin{definition}
	Let $G=(V,E)$ be a simple triangle-free graph with maximum degree at most $3$, and $F\subset E$. Then, the \emph{forced cubic Hamiltonian cycle} (FCHC) problem asks whether $G$ has a Hamiltonian cycle containing all edges in $F$. We call edges in $F$ \emph{forced}, and edges in $E\backslash F$ \emph{unforced}. We call $(G,F)$ an \emph{FCHC instance}. 
\end{definition}

 Roughly speaking, Eppstein's algorithm solves FCHC by recursively selecting an unforced edge and creating two subinstances by either adding that edge to $F$ or removing it from $G$. In both cases, the fact that $G$ is cubic induces additional edges to be either added to $F$ or to be removed. 

The details of Eppstein's algorithm are given in Alg.~\ref{alg:eppsteinClassical}. The  formulation of the algorithm here has been modified from Eppstein's original formulation in several places. In particular, we  adapted it to the Hamiltonian cycle  problem (instead of the travelling salesman problem\footnote{The result can be generalised to the travelling salesman problem, subject to constraints on the distances depending on the number of available qubits.}) and solve it as a decision problem (rather than finding a Hamiltonian cycle). We also made several smaller changes to make the transition the the quantum algorithm later easier. 
For clarity, and to make this section self-contained, $\Eppstein$ will in the following always refer to the algorithm in Alg.~\ref{alg:eppsteinClassical} instead of the original formulation of this algorithm in \cite{Eppstein}. 

\begin{algorithm}
\colorbox{algbgcolour}{\parbox{\textwidth}{\vspace{-8pt}\begin{flushleft}
$\textsc{Eppstein}(G,F)$:
\end{flushleft}  
\vspace{-10pt}
\begin{enumerate}
	\item\label{enum:classicalEppsteinStep1} Repeat the following steps (``trivial reductions'') until none of the conditions apply
	\begin{enumerate}[a.]
		\item\label{enum:classicalEppsteinTrivialDeg2} If $G$ contains a vertex with degree two with at least one unforced incident edge, add all its incident edges to $F$.
		\item\label{enum:classicalEppsteinTrivialDeg3} If $G$ contains a vertex with degree three with exactly two forced edges, remove the unforced edge. 
		\item\label{enum:classicalEppsteinTrivialCycle} If $G$ contains a cycle of four unforced edges such that two of its opposite vertices are each incident to a forced edge and at least one of the other vertices is incident to an unforced edge that is not part of the cycle, then add to $F$ all non-cycle edges that are incident to a vertex of the cycle.
	\end{enumerate}
	\item\label{enum:classicalEppsteinStep2}  Check if any of the following conditions (``terminal conditions'') apply
	\begin{enumerate}[a.]
		\item\label{enum:classicalEppsteinStep2a} If $G$ contains a vertex of degree $0$ or $1$, or if $F$ contains three edges meeting at a vertex, return \emph{false}.
		\item\label{enum:classicalEppsteinStep2b} If $G\backslash F$ is a collection of disjoint $4$-cycles and isolated vertices
		\begin{enumerate}[i.]
				\item If $G$ is disconnected, return \emph{false}.
				\item Otherwise, return \emph{true}.
		\end{enumerate}		 
		\item\label{enum:classicalEppsteinTerminalUnneccessary} If $F$ contains a non-Hamiltonian cycle, return \emph{false}.
	\end{enumerate}
	\item \label{enum:classicalEppsteinStep3} Choose an edge $yz$ according to the following cases:
	\begin{enumerate}[a.]
		\item \label{enum:classicalEppsteinStep3a} If $G\backslash F$ contains a $4$-cycle, exactly two vertices of which 
	are incident to an edge in $F$, 		
		let $y$ be one of the other two vertices of the cycle and let $yz$ be an edge of $G\backslash F$ that does not belong to the cycle.
		\item \label{enum:classicalEppsteinStep3b} If there is no such $4$-cycle, but $F$ is nonempty, 
		let $xy$ be any edge in $F$ and $yz$ be an adjacent edge in $G\backslash F$ such that $yz$ is not part of an isolated $4$-cycle in $G\backslash F$. 
		\item \label{enum:classicalEppsteinStep3c} Otherwise, let $yz$ be any edge in $G$ that is not part of an isolated $4$-cycle in $G\backslash F$.  
	\end{enumerate}
	\item\label{enum:classicalEppsteinStep4} Call   $\textsc{Eppstein}(G, F\cup \{yz\})$.
	\item\label{enum:classicalEppsteinStep5} Call  $\textsc{Eppstein}(G\backslash \{yz\}, F)$. 
	\item\label{enum:classicalEppsteinStep6} Return the disjunction (logical OR) of steps~\ref{enum:classicalEppsteinStep4} and \ref{enum:classicalEppsteinStep5}. 
\end{enumerate}}}
\caption{Eppstein's algorithm (modified)}\label{alg:eppsteinClassical}
\end{algorithm}

We first introduce a few important concepts.

\begin{definition}
	An FCHC instance $(G,F)$ is called \emph{trivial-reduction-free} if 
	\begin{enumerate}[(i)]
		\item $G$ does not contain any vertices of degree two with unforced incident edges,
		\item $G$ does not contain any vertices of degree three with exactly two forced edges, and
		\item $G$ does not contain any cycles of four unforced edges such that two of its opposite vertices are incident to a forced edge and at least one of the other vertices is incident to an unforced edge that is not part of the cycle.
	\end{enumerate}
\end{definition}

In other words, an FCHC instance is trivial-reduction-free if and only if none of the conditions of step~\ref{enum:classicalEppsteinStep1} of $\Eppstein$ apply. 

\begin{definition}
	Let $G=(V,E)$ be a simple, triangle-free graph with maximum degree at most $3$,  $\omega$ be a cycle of four edges, and $F\subset E$. We say that $\omega$ \emph{unforced-isolated w.r.t. $F$} if all edges of $\omega$ are in $E\backslash F$, and each of the four vertices of $\omega$ is incident to an edge in $F$. Moreover, denote by $C(G,F)$ the set of $4-$cycles in $G$ which are unforced-isolated w.r.t. $F$. 
\end{definition}

Note that if $(G,F)$ is trivial-reduction-free, then all unforced $4$-cycles are unforced-isolated w.r.t. $F$. 
The correctness of Alg.~\ref{alg:eppsteinClassical} is given by the following proposition. 

\begin{proposition}\label{prop:correctness}
	Let $(G,F)$ be a trivial-reduction-free FCHC instance. Suppose that $G$ has only vertices of degree $2$ or $3$, and that no three edges in $F$ meet in a single vertex. Suppose moreover that $G\backslash F$ is a collection of disjoint $4$-cycles and isolated vertices. Then, $G$ has a Hamiltonian cycle containing all edges in $F$ if and only if $G$ is connected. 
\end{proposition}
\begin{proof}
The ``only if'' direction is trivial. Assume that $G=(V,E)$ is connected.  
	Note first of all that all vertices outside of $C(G,F)$ have degree $2$ and both their incident edges are in $F$.
	
	For each $4$-cycle $\omega\in C(G,F)$, let $h_1(\omega),h_2(\omega)$ be two opposite edges in $\omega$ and $h_3(\omega), h_4(\omega)$ be the other two edges in $\omega$. Let $F_1 := F\cup \{ h_1(\omega), h_2(\omega) : \omega\in C(G,F)\}$. Note that every vertex in $G$ is incident to an edge in $F_1$ and that every vertex in $G_1=(V,F_1)$ has degree two. Thus, $G_1$ is a collection of cycles. Consider the  graph $G_2$ whose vertices are the connected components of $G_1$, and  two connected components $H_1,H_2$ of $G_1$ are joined by an edge in $G_2$ iff there exists a $4$-cycle $\omega\in C(G,F)$ which ``separates'' $H_1,H_2$ in $G$, i.e., $h_1(\omega)\in H_1$ and $h_2(\omega)\in H_2$ or vice-versa. Note that since $G$ is connected, so is $G_2$. Note moreover that if $H_1$ and $H_2$ are adjacent in $G_2$ and separated by $\omega\in C(G,F)$ in $G$, replacing $h_1(\omega),h_2(\omega)$ with $h_3(\omega),h_4(\omega)$ in $G_1$ would result in replacing the two disconnected cycles $H_1$, $H_2$ by a single cycle going through the same vertices. Thus, consider a spanning tree of $G_2$ and let $G_1'=(V,F_1')$ be the graph obtained from $G_1$ by replacing $h_1(\omega),h_2(\omega)$ with $h_3(\omega),h_4(\omega)$ for all $\omega$ corresponding to an edge in the spanning tree. Then, by the previous observation, $G_1'$ is connected. Moreover, since $F\subset F_1'$ and every vertex has degree $2$ in $G_1'$, it follows that $G_1'$ is a Hamiltonian cycle containing all edges in $F$. 
\end{proof}

Note in particular that Proposition~\ref{prop:correctness} implies that step~\ref{enum:classicalEppsteinTerminalUnneccessary} of $\Eppstein$ is in fact unneccessary and the algorithm still performs correctly if that step is omitted. We include that step nevertheless, since the early termination of these instances simplifies the runtime analysis in Appendix~\ref{app:runtime}.

The main idea of bounding the runime of $\Eppstein$ is to introduce a ``problem size metric'' defined as follows.

\begin{definition}
 For an FCHC instance $(G,F)$, let $s(G,F):= \max(n(G)-|F|-|C(G,F)|, 0)$. We call $s(G,F)$ the \emph{size} of $(G,F)$.\end{definition}

\begin{proposition}\label{prop:snonnegative}
	Let $(G,F)$ be an FCHC instance such that no three edges in $F$ meet at a vertex. Suppose moreover that $F$ is not a collection of cycles. Then, $n(G)-|F|-|C(G,F)|> 0$.
\end{proposition}
\begin{proof}
	First of all, note that the general case can be reduced to the special case of $C(G,F)=\emptyset$ by adding one edge of each $\omega\in C(G,F)$ to $F$. Suppose now that $C(G,F)=\emptyset$, and let $G'=(V,F)$, where $V$ is the set of vertices of $G$. Then, every vertex has degree at most $2$ in $G'$. Moreover, since $F$ is not a collection of cycles, at least one vertex has degree $<2$ in $G'$. Hence,  
	\begin{equation}
		2|F| = \sum_{v\in V} \deg_{G'} v < 2|V|
	\end{equation}
	and hence $n(G)-|F|> 0$. 
\end{proof}

  It can be shown \cite{Eppstein} that every application of steps~\ref{enum:classicalEppsteinStep3}--\ref{enum:classicalEppsteinStep6} creates two FCHC instances $(G_1,F_1)$ and $(G_2,F_2)$ such that $s(G_1,F_1), s(G_2,F_2) \leq s(G,F) - 3$ or $s(G_1,F_1) \leq s(G,F)-2$ and $s(G_2,F_2) \leq s(G,F)-5$. This leads to a recursive runtime bound of $T(s(G,F))$, where $T(s) = \max(2T(s-3), T(s-2)+T(s-5)) + O(\poly(n(G)))$, leading to $T(s) = O(2^{s/3}\poly(n(G)))$.  
  In particular,  $\Eppstein(G,\emptyset)$ solves the 
  cubic Hamiltonian cycle problem in a runtime of $O(2^{n(G)/3}\poly(n(G)))$. 
  We provide the details of this runtime analysis in the Appendix~\ref{app:runtime}.

\subsection{Quantum improvement of $\Eppstein$ using small quantum computer}

Using Theorem~\ref{thm:sha}, the main part of obtaining a speedup is to provide a quantum speedup of $\Eppstein$ using few qubits. 

\begin{theorem}\label{thm:quantumAlg}
	There exists a quantum algorithm that, for any FCHC instance $(G,F)$, decides FCHC in a runtime of $O(2^{s/4}\poly(n))$ using $O(s\log(n/s) + s+\log n)$ qubits, where $s=s(G,F)$ and $n=n(G)$. 
\end{theorem}

We will prove this in Section~\ref{sec:proofQAlg}. Theorems~\ref{thm:sha} and \ref{thm:quantumAlg} immediately imply an improvement to Eppstein's algorithm.

\begin{theorem} \label{thm:eppsteinSpeedup}
Let $c>0$ be an arbitrary constant. Then, given a quantum computer with $M=cn$ qubits, there exists a hybrid quantum-classical algorithm that solves the cubic Hamiltonian cycle problem for $n$-vertex graphs in runtime $O(2^{(1/3-f(c))n}\poly(n))$, where $f(c)>0$. 
\end{theorem}
Note that with Theorem~\ref{thm:quantumAlg}, we have a quantum algorithm that satisfies all the criteria to apply Theorem~\ref{thm:sha}. However, to make things fully rigorous, we also need to deal with the fact that Alg.~\ref{alg:eppsteinClassical} includes initial rewritings and reductions, and as such is not immediately a special case of Alg.~\ref{alg:algsha}. This is a minor technicality, which we resolve as follows for completeness. 
\begin{proof}[Proof of Theorem~\ref{thm:eppsteinSpeedup}]
	Note first of all that it is sufficient to prove that one can solve the FCHC problem in a runtime of $O(\max(2^{\gamma s(G,F) - f(c)n(G)},2^{\gamma_Q s(G,F)})\poly(n(G)))$, where $\gamma=1/3$ and $\gamma_Q=1/4$. Next, note that it is sufficient to do this only for trivial-reduction-free FCHC instances, since the others can be reduced to the trivial-reduction-free case by one application of step~\ref{enum:classicalEppsteinStep1} of $\Eppstein$. 
	
	Thus, with the notation of Section~\ref{sec:sha}, let $\mathcal P$ be the set of all trivial-reduction-free FCHC instances. Then $\alg$ is given as follows: $\trivial(P)=1$ if any of the terminal conditions in step~\ref{enum:classicalEppsteinStep2} of $\Eppstein$ apply, and $f(P)$ is the value returned in step~\ref{enum:classicalEppsteinStep2}. Moreover, $l=2$ and  $R_{1}$ and $R_2$ are given by first selecting an edge $yz$ according to step~\ref{enum:classicalEppsteinStep3} of $\Eppstein$, forcing and deleting it, respectively, followed by performing all possible trivial reductions (i.e., step~\ref{enum:classicalEppsteinStep1} of $\Eppstein$). Note that indeed, $R_1$ and $R_2$ map trivial-reduction-free FCHC instances to trivial-reduction-free FCHC instances. Finally, the runtime analysis of $\Eppstein$ (see Proposition~\ref{prop:sbound} in Appendix~\ref{app:runtime}) gives $k=3$, $l_1=l_2=l_3=2$ and 
	\begin{equation}
		C = \left(
			\begin{array}{ccc}
				3 & 2 & 5 \\
				3 & 5 & 2 
			\end{array}
		\right). 
	\end{equation}
	The result now follows from Theorem~\ref{thm:sha}. 
\end{proof}

\subsection{Proof of Theorem~\ref{thm:quantumAlg}} \label{sec:proofQAlg}

To prove Theorem~\ref{thm:quantumAlg}, note first of all that we can without loss of generality assume that $(G,F)$ is trivial-reduction-free, because if it is not, $G$ and $F$ can be classically pre-processed by repeated applications of step~\ref{enum:classicalEppsteinStep1} of $\Eppstein$. Note that $s(G,F)$ is non-increasing under this step. 

 The recursive steps~\ref{enum:classicalEppsteinStep4} and \ref{enum:classicalEppsteinStep5} in $\Eppstein$ yield a binary recursion tree, with the two branches corresponding to either step~\ref{enum:classicalEppsteinStep4} or \ref{enum:classicalEppsteinStep5}. 
More formally, consider the binary rooted tree defined as follows. 

\begin{definition}
 Let $(G,F)$ be an FCHC instance. Define $\rectree(G,F)$ to the binary rooted tree constructed as follows: the root of the tree is a vertex labelled $(G,F)$. Then, given a vertex $(\tilde G, \tilde F)$, if $\Eppstein(\tilde G, \tilde F)$ terminates in step~\ref{enum:classicalEppsteinStep2}, it becomes a leaf of the tree. Otherwise, if $\Eppstein(\tilde G,\tilde F)$ calls $\Eppstein(G_1,F_1)$ and $\Eppstein(G_2,F_2)$ in steps~\ref{enum:classicalEppsteinStep4} and \ref{enum:classicalEppsteinStep5}, respectively, create two children of $(\tilde G,\tilde F)$ labelled $(G_1,F_1)$ and $(G_2,F_2)$, respectively. 
 
 Moreover, let $\tau( G, F)$ be the number of edges that are forced or deleted in step~\ref{enum:classicalEppsteinStep1} of $\Eppstein( G, F)$, before the algorithm moves on. 
\end{definition}

Intiutively, $\rectree(G,F)$ is the tree of FCHC instances explored by steps~\ref{enum:classicalEppsteinStep4} and \ref{enum:classicalEppsteinStep5} of Alg.~\ref{alg:eppsteinClassical}. 
Note that by Proposition~\ref{prop:sbound} in Appendix~\ref{app:runtime}, 
$\rectree(G,F)$ has depth at most $s(G,F)/2$. 
We first show that in a ``good'' branch of the recursion tree, a total of at most $O(s(G,F))$ trivial reductions are performed.

\begin{proposition}\label{prop:trivialreductionbound}
	Let $(G,F)$ be trivial-reduction-free and let $(G,F)=(G_1,F_1), (G_2,F_2),\ldots,(G_l,F_l)$ be a path in $\rectree(G,F)$, such that $\Eppstein(G_l,F_l)$ returns true in step~\ref{enum:classicalEppsteinStep2}. Then, 
	\begin{equation}
		\sum_{j=1}^l \tau(G_j,F_j) \leq 4s(G,F).
	\end{equation}
\end{proposition}
\begin{proof}
	First, note that no edge which is part of some $\omega \in C(G,F)$ will be forced at any point of the algorithm. On the other hand, every Hamiltonian cycle contains exactly two (opposite) edges of each $\omega \in C(G,F)$. Hence, along the path from $(G_1,F_1)$ to $(G_l,F_l)$, 
a total of at most $n(G)-|F|-2|C(G,F)|\leq s(G,F)$ edges will be forced. In particular, at most $s(G,F)$ edges will be forced in step~\ref{enum:classicalEppsteinTrivialDeg2} or \ref{enum:classicalEppsteinTrivialCycle}	of $\Eppstein(G_j,F_j)$ over all $j=1,\ldots,l$. 
	As for the number of edges deleted in step~\ref{enum:classicalEppsteinTrivialDeg3}, note that 
	since $(G,F)$ is trivial-reduction-free, 
	every deletion of an edge in step~\ref{enum:classicalEppsteinTrivialDeg3} of $\Eppstein(G_j,F_j)$ for some $j \in\{1,\ldots,l\}$  is induced by an additional edge being forced. Each such edge can induce at most two edges being deleted in step~\ref{enum:classicalEppsteinTrivialDeg3}. Therefore, at most $2(s(G,F)+ s(G,F)/2) =3s(G,F)$ edges are deleted in step~\ref{enum:classicalEppsteinTrivialDeg3} of  $\Eppstein(G_j,F_j)$ over all $j=1,\ldots,l$. 
\end{proof}

The quantum algorithm we construct is essentially a quantum version of a non-recursive variant of $\Eppstein$ which proceeds by successively forcing and deleting edges according to a given input of a suitable search space. 
Note that $G$ and $F$ are classical input parameters and as such, the quantum circuit may depend on $G$ and $F$. We will not actively modify $G$ or $F$ in the quantum algorithm. Instead, the removal and forcing of additional edges are done by quantumly storing encodings of a set $X\subset\{1,\ldots, 3|E|\}$, where for $e\in\{1,\ldots,|E|\}$, $e\in X$ means that edge number $e$ is forced, and $e+|E|\in X$ means that edge number $e$ has been removed (we assume that all edges and vertices are enumerated in a pre-specified order, i.e., with some abuse of notation\footnote{For simplicity of notation, we will in this section not distinguish between an edge $e$ and its number in the enumeration, and simply write $e=10$ to mean that $e$ is the $10$\textsuperscript{th} edge according to the enumeration. We do the same for vertices.}, $E=\{1,\ldots,|E|\}$ and $V=\{1,\ldots, |V|\}$). 
For convenience we  also introduce dummy variables which correspond to the values $2|E|+1,\ldots,3|E|$. 

\begin{algorithm}[htpb]
\colorbox{algbgcolour}{\parbox{\textwidth}{\begin{algorithmic}[1]
\Procedure{NonRecursiveEppstein}{$G,F$}
	\State 
	$s := n(G)-|F|-|C(G,F)|$, 
	 $r:=\lfloor s/2\rfloor + 4s$
	\ForAll{$\vec \nu \in \{0,1\}^r$}
		\State $X := \reduce(G,F,\vec \nu)$
		\State $h:= \checkcycles(G,F,X)$
		\If{$h=1$}
			\State\Return true
		\EndIf
	\EndFor	
	\State\Return false
\EndProcedure
\State
\Procedure{$\reduce$}{$G, F, \vec \nu = (\nu_1,\ldots,\nu_r)$}
	\State $X:=\emptyset$ 
	\For{$i=1,\ldots,r$ } 
		\State $x:=\textsc{Calculate}(G,F,X,\nu_i)$
		\State $X := X\cup \{x\}$
	\EndFor
 	\State \Return $X$
\EndProcedure
\State 
\Procedure{$\calculate$}{$G=(V,E), F, X, \nu$}
	\State $F':= \{e\in\{1,\ldots,|E| : e\in X\}$,
	 $D := \{e\in\{1,\ldots,|E| : e+|E|\in X\}$,
	 $G':=G\backslash D$ 
	\If {$G'$ contains a vertex with degree $2$ with at least one edge in $E\backslash (F\cup F' \cup D)$} \label{alg2case1-1}
		\State $e := $  one of the  edges in $E\backslash (F\cup F'\cup D)$ incident to that vertex, \label{alg2case1-2}
		 $a := 0$
	\ElsIf {$G'$ contains a vertex with degree $3$ with exactly two edges in $F\cup F'$}\label{alg2case2-1}
		\State $e :=$ the third edge incident to that vertex,
		 $a :=1$\label{alg2case2-2}
	\ElsIf {$G'\backslash (F\cup F')$ contains a cycle of $4$ edges with two of its opposite vertices being incident to an edge in $F\cup F'$ and one of the  other two vertices being incident to a non-cycle edge in $E\backslash (F\cup F'\cup D)$}\label{alg2case3-1}
		\State $e:=$ that non-cycle edge in $E\backslash (F\cup F'\cup D)$,
		 $a := 0$\label{alg2case3-2}
	\ElsIf {$G'$ contains a vertex of degree $0$ or $1$, or if $F\cup F'$ contains three edges meeting at a vertex, or if $G'\backslash (F\cup F') $ is a collection of disjoint $4$-cycles and isolated vertices}\label{alg2case4-1}
		\State $e := i+2|E|$,
		 $a := 0$\label{alg2case4-2}
	\ElsIf {$G'\backslash (F\cup F')$ contains a $4$-cycle, two vertices of which  
	are incident to an edge in $F\cup F'$} 	
	\label{alg2case5-1} 
		\State $e:=$ an edge in $G'\backslash (F\cup F') $ incident to one of the other two vertices and which does not belong to the cycle,
		 $a := \nu$\label{alg2case5-2}
	\ElsIf {$F\cup F'$ is nonempty and not a collection of cycles}\label{alg2case6-1}
		 \State $e:=$ any edge in $E\backslash (F\cup F'\cup D)$ that is adjacent to an edge in $F\cup F'$,
		 $a:= \nu$\label{alg2case6-2}
	\Else\label{alg2case7-1}
		\State $e:=$ any edge in $G'$ that is not part of a $4$-cycle in $G'\backslash (F\cup F')$,
		 $a:=\nu$ \label{alg2case7-2}
	\EndIf
	\State\Return $e+a|E|$
\EndProcedure
\State
\Procedure{$\checkcycles$}{$G=(V,E),F,X$}
  	\State $F':= \{e\in\{1,\ldots,|E| : e\in X\}$,
	 $D := \{e\in\{1,\ldots,|E| : e+|E|\in X\}$,
	 $G' := (V, E\backslash D)$
	\If{$G'$ contains a vertex of degree 1} \label{algcheckcase1}
		\State\Return false
	\ElsIf {$F\cup F'$ contains three edges meeting at a vertex}\label{algcheckcase2}
		\State\Return false
	\ElsIf {$G\backslash (F\cup F') $ is not a collection of disjoint $4$-cycles and isolated vertices}\label{algcheckcase3}
		\State\Return false
	\ElsIf {$G$ is disconnected}\label{algcheckcase4}
		\State\Return false
	\EndIf 
	\State\Return true
\EndProcedure
\end{algorithmic}}}
\caption{Non-recursive variant of $\Eppstein$. Note that $G$ and $F$ are never modified, and that in the $\reduce$ subroutine, $|X|=i$ after each cycle of the loop.}\label{alg:nonrecursiveEppstein}
\end{algorithm}

Before moving to the quantum algorithm, we first show how to perform $\Eppstein$ non-recursively. 
By Proposition~\ref{prop:trivialreductionbound} and Proposition~\ref{prop:sbound} in Appendix~\ref{app:runtime}, a given path leading to an accepting leaf  of the recursion tree $\rectree(G,F)$  adds at most $r:=\lfloor s(G,F)/2\rfloor + 4s(G,F)$ elements to $X$, corresponding to at most $4s(G,F)$ edges deleted or forced in step~\ref{enum:classicalEppsteinStep1} (Proposition~\ref{prop:trivialreductionbound}) of Alg.~\ref{alg:eppsteinClassical}, and at most $s(G,F)/2$ in steps~\ref{enum:classicalEppsteinStep4} or \ref{enum:classicalEppsteinStep5} (Proposition~\ref{prop:sbound}). Thus, a sequence of  $r$ binary variables $\nu_1,\ldots, \nu_r\in\{0,1\}$ suffice to enumerate all relevant leafs of $\rectree(G,F)$, 
where each variable $\nu_i$ corresponds to an edge deleted or forced in step~\ref{enum:classicalEppsteinStep1}, step~\ref{enum:classicalEppsteinStep4} or step~\ref{enum:classicalEppsteinStep5}. 
More precisely, 
the value of $\nu_i$ is ignored if the $i$\textsuperscript{th} element added to $X$ is through a trivial reduction (i.e., in step~\ref{enum:classicalEppsteinStep1} of $\Eppstein$), otherwise the value of $\nu_i$ specifies if the edge chosen in step~\ref{enum:classicalEppsteinStep3} is forced (step~\ref{enum:classicalEppsteinStep4}) or removed (step~\ref{enum:classicalEppsteinStep5}). 
Thus, $\Eppstein$ induces a mapping from $\mathcal I:=\{0,1\}^r\rightarrow \mathcal P_r([3|E|])$, $\vec \nu:=(\nu_1,\ldots, \nu_r)\mapsto X=X(\vec \nu)$. 

Note that although the same task could be achieved by only introducing $\lfloor s(G,F)/2\rfloor$ binary variables instead of $r$, the additional $4s(G,F)$ variables ensure that at any given point of the implementation, $X$ has a pre-determined size, which would not be the case otherwise (indeed,  Proposition~\ref{prop:trivialreductionbound} only provides an upper bound to the number of trivial reductions overall but these are generally distributed in a previously unknown way). This will be important when we later want to use the results of Section~\ref{sec:iterative}.  

The non-recursive (classical) variant of $\Eppstein$ is given in Alg.~\ref{alg:nonrecursiveEppstein}. Write $X=F'\cup D$, where $F'=\{e\in\{1,\ldots, |E|\} : e\in X\}$ and $D=\{e\in\{1,\ldots, |E|\} : e+|E|\in X\}$ are the edges which would be forced and deleted in either of steps~\ref{enum:classicalEppsteinStep1}, \ref{enum:classicalEppsteinStep4} or \ref{enum:classicalEppsteinStep5} of $\Eppstein(G,F)$, respectively. 
The algorithm first computes $X(\vec \nu)$ from $\vec \nu$ (we call this operation $\reduce$), and then checks 
the conditions of step~\ref{enum:classicalEppsteinStep2} for the FCHC instance $(G\backslash D, F\cup F')$ and returns the corresponding value. 
We call the second step $\checkcycles$. For completeness, $\checkcycles$ also returns false if none of the conditions of step~\ref{enum:classicalEppsteinStep2} apply (note that if this is the case, Propostions~\ref{prop:correctness} and \ref{prop:trivialreductionbound} imply that  $\vec\nu$ corresponds to a branch that does not find a Hamiltonian cycle). The full non-recursive version $\Eppstein$ simply goes through all $2^{r}$ values of $\vec \nu$, yielding a runtime of $O(2^{r}\poly(n(G)))$. Alternatively, by picking $\vec \nu\in\mathcal I$ uniformly at random, an expected rutnime of $O(2^r/2^t\poly(n(G))) = O(2^{s/2}\poly(n(G)))$ can be achieved. 

Note that we have omitted checking the conditions of step~\ref{enum:classicalEppsteinTerminalUnneccessary} of $\Eppstein$ in the non-recursive formulation, as it does not affect the correctness of the algorithm (see also remark after Proposition~\ref{prop:correctness}). Indeed, for a  ``good'' branch (i.e., a value of $\vec\nu$ that corresponds to finding a Hamiltonian cycle), the condition of step~\ref{enum:classicalEppsteinTerminalUnneccessary} of $\Eppstein$ is never fulfilled, whereas Proposition~\ref{prop:correctness} ensures that if $F\cup F'$ contains a non-Hamiltonian cycle, $\checkcycles$ will return false. Instead of checking the condition of step~\ref{enum:classicalEppsteinTerminalUnneccessary} of $\Eppstein$, the condition in line~\ref{alg2case6-1} of Alg.~\ref{alg:nonrecursiveEppstein} has been modified to account for the case when $F$ is a collection of (non-Hamiltonian) cycles, in which case step~\ref{enum:classicalEppsteinStep3b} of $\Eppstein$ cannot be applied. 

To turn Alg.~\ref{alg:nonrecursiveEppstein} into a quantum algorithm, we first show that $\reduce$ and $\checkcycles$ can be performed \emph{reversibly} in polynomial time, and with only $O(s\log(n/s)+s+\log n)$ bits, where $s=s(G,F)$ and $n=n(G)$, and we assume that $\checkcycles$ simply writes the result on a single output bit. Note that throughout $\textsc{NonRecursiveEppstein}(G,F)$, $G$ and $F$ are never modified. As such, it is sufficient to write $\reduce$ and $\checkcycles$ as reversible circuits $\reduce_{G,F}$ and $\checkcycles_{G,F}$, which depend on $G$ and $F$, and which perform
\begin{equation}\label{eq:reducecheckcyclesdefine}
	\ket{\vec \nu}\ket 0\ket 0\ket 0 \ket 0 
		\stackrel{\reduce_{G,F}}{\longrightarrow}
	\ket{\vec \nu}\ket {\EffEnc Z(\vec \nu)}\ket0\ket0
		\stackrel{\checkcycles_{G,F}}{\longrightarrow}
	\ket{\vec \nu}\ket {\EffEnc Z(\vec \nu)}\ket {h(Z(\vec \nu))} \ket 0,
\end{equation}
where $Z(\vec \nu) = (x_1,\ldots,x_r)$, $x_1,\ldots,x_r$ are the elements of $X(\vec\nu)$ in the order in which they are added to $X$ in $\reduce$,  $h(Z(\vec \nu))$ is the output bit returned by $\checkcycles(G,F,X(\vec\nu))$, and the last register comprises at most $O(s\log(n/s)+s+\log n)$ ancilla bits. 

 Next, we turn this into a quantum process by replacing every elementary reversible operation by its corresponding quantum operation. 
The final step is to use amplitude amplification \cite{2000_Brassard} (or alternatively fixed point search \cite{YoderLowChuang14}) to  quantumly search for the value $\ket{h(Z(\vec\nu))}=\ket 1$ on the output register. Note that since $t$ of the $r$ input bits are irrelevant, the dimension of the target space, if a Hamiltonian cycle exists, is at least $2^t$. Thus, fixed point search requires $O(\sqrt{2^r/2^t})=O(2^{s/4})$ repetitions of $\reduce_{G,F}$ and $\checkcycles_{G,F}$.

It thus only remains to show that for any given trivial-reduction-free FCHC instance $(G,F)$,  both $\reduce_{G,F}$ and $\checkcycles_{G,F}$ can be performed reversibly in polynomial time and with only $O(s\log(n/s)+s+\log n)$ bits, where $s=s(G,F)$ and $n=n(G)$.

\subsubsection{Reversible space-efficient implementation of $\reduce_{G,F}$}

The basic idea of the $\reduce_{G,F}$ algorithm is to reversibly generate an efficient encoding of $X(\vec \nu)\subset\{1,\ldots, 3|E|\}$ from $\ket{\vec \nu}$ using Theorem~\ref{thm:setgen} by specifying suitable reversible $\calculate_i$ operations. For each value of $i=1\ldots, r$, let $\calculate_{G,F,i}$ be the operation that finds the next edge $e_i$ and an \emph{action bit} $a_i\in \{0,1\}$ following the implementation of $\textsc{Calculate}$ in Alg.~\ref{alg:nonrecursiveEppstein}. That implementation in turn follows the implementation of $\Eppstein$, i.e., an edge $e_i$ is selected according to the rules of step~\ref{enum:classicalEppsteinStep1} and \ref{enum:classicalEppsteinStep3}, and  $a_i=0$ if $e_i$ will be forced and $a_i=1$ if $e_i$ will be removed.  
If no more edges are forced or deleted (i.e., one of the terminal conditions of step~\ref{enum:classicalEppsteinStep2a} or \ref{enum:classicalEppsteinStep2b} of $\Eppstein$ apply), we set $e_i$ to be a dummy variable $e_i=2|E|+i$ and $a_i=0$.

The task then becomes to reversibly implement the operations
\begin{equation}\label{eq:calculate1def}
	\calculate_{G,F,1} \ket{\nu_1} \ket 0= \ket{\nu_1}\ket{x_1}
\end{equation}
and
\begin{equation}\label{eq:calculateidef}
	\calculate_{G,F,i} \ket{\nu_i} \ket{\EffEnc Z_{i-1}}\ket 0= \ket{\nu_i}\ket{\EffEnc Z_{i-1}}\ket{x_i}
\end{equation}
for $i=2,\ldots,r$, 
where $x_i=e_i+a_i|E|$ and $Z_i = (x_1,\ldots,x_i)$. Note that by introducing the dummy variables, we ensure that for any $i\in\{1,\ldots,r\}$, $X_i=\{x_1,\ldots,x_i\}$ has exactly $i$ elements, even if no edges have been forced or deleted in some of the steps.

The primary challenge of the implementation is to maintain reversibility and at the same time use few ancillas. For this, it is important that any ancillas used are as soon as possible reset to $\ket 0$ in order to avoid accumulating unnecessary junk bits.  

\begin{proposition}\label{prop:calculate}
 Let $(G,F)$ be a trivial-reduction-free FCHC instance and $i\in\{1,\ldots,r\}$, where $r=\lfloor s(G,F)/2\rfloor + 4s(G,F)$. Then, the operation $\calculate_{G,F,i}$ defined by Eq.~\eqref{eq:calculate1def} and \eqref{eq:calculateidef} can be implemented reversibly using $O(\log n(G))$ ancillas and $O(\poly(n(G)))$ gates.
\end{proposition}
\begin{proof}
The basic idea to maintain reversibility is to introduce a counter, initially set to $0$, which is increased once suitable values for $e_i$ and $a_i$ have been found. Then, by controlling all operations on that counter being $0$, we ensure that no further edges are selected once an edge and action bit has been found, and hence that only one edge and action bit is selected. 

More precisely, we introduce a counter from $0$ to $7$, which we call the \emph{flag counter}, and denote it by $\mathcal{FC}$. We also introduce an additional ancilla bit, which we call the \emph{flag bit}, and denote it by $\mathcal{FB}$. Both are initially set to zero. We denote the register of $O(\log n)$ bits containing the value of $e_i$ as $\mathfrak e$ and the register containing the value of $a_i$ as $\mathfrak a$. We assume that $\mathfrak e$ and $\mathfrak a$ are both initially set to zero. 

For clarity of notation, we will only cover the case $i\geq 2$ here. The case $i=1$ is fully analogous, but without the $\ket{\EffEnc Z_{i-1}}$ register and ignoring any operations involving it. As before, we write  $Z_{i-1}=(x_1,\ldots,x_{i-1})$, $X_{i-1}=\{x_1,\ldots,x_{i-1}\}$, $F'=\{e\in\{1,\ldots,|E|\}: e\in X_{i-1}\}$ and $D=\{e\in\{1,\ldots,|E|\}: e+|E|\in X_{i-1}\}$.

\begin{figure}[htbp]
	\begin{tikzpicture}
	
	\node at (0,0) {\begin{tikzpicture}[scale=.8]
	
		\draw (0,-.5) -- (7,-.5); 
		\draw (0,1) -- (7,1);
		\draw (0,2) -- (7,2);
		\draw (0,3) -- (7,3);
		\draw (0,4) -- (7,4);
		\draw (0,5) -- (7,5);
		\draw (0,6) -- (7,6);
				
		\draw (3,1) -- (3,-.5);
		
		\draw (6,2) -- (6,1.1);
		\draw (6,.9) -- (6,-.5);			
		
		\filldraw[fill=red!10, draw] (1,.5) rectangle (5,6.5);		
		
		\filldraw[fill=green!10, draw] (3,-.5) circle (.5);
		\node at (3,-.5) {$=\!0?$};
		
		\filldraw[fill=blue!10, draw] (5.5,0) rectangle (6.5,-1);
		\node at (6,-.5) {$+1$};
		
		\node[circle,inner sep=2pt,fill] at (6,2) {};
		
		\node[left] at (0,-.5) {$\mathcal{FC}$};
		
		\node[left] at (0,4) {$\mathfrak e$};
		\node[left] at (0,3) {$\mathfrak a$};
		
		\node[left] at (0,2) {$\mathcal{FB}$};
		
		\node[left] at (0,1) {$\ket{0}_{\text{anc}}$};
		\node[right] at (7,1) {$\ket{0}_{\text{anc}}$};
		
		\node[left] at (0,6) {$\ket{\nu_i}$};
		\node[right] at (7,6) {$\ket{\nu_i}$};

		\node[left] at (0,5) {$\ket{\EffEnc Z_{i-1}}$};
		\node[right] at (7,5) {$\ket{\EffEnc Z_{i-1}}$};

		\node at (3,3.5) {\large Check \& Select};

	\end{tikzpicture}};
	
	\node at (6,0) {$=:$};
	
	\node at (8.5,0) {\begin{tikzpicture}[scale=.8]
	
		\draw (0,-.5) -- (4,-.5); 
		\draw (0,1) -- (4,1);
		\draw (0,2) -- (4,2);
		\draw (0,3) -- (4,3);
		\draw (0,4) -- (4,4);
		\draw (0,5) -- (4,5);
		\draw (0,6) -- (4,6);		
		
		\filldraw[fill=orange!10, draw] (1,-1) rectangle (3,6.5);						
		\node at (2,3) {$\text{CCS}$};

	\end{tikzpicture}};
	
	\end{tikzpicture}
	\caption{Controlled check-and-select (CCS) operation. The ancilla register consists of $O(\log n)$ bits.}\label{fig:checkselect}
\end{figure}
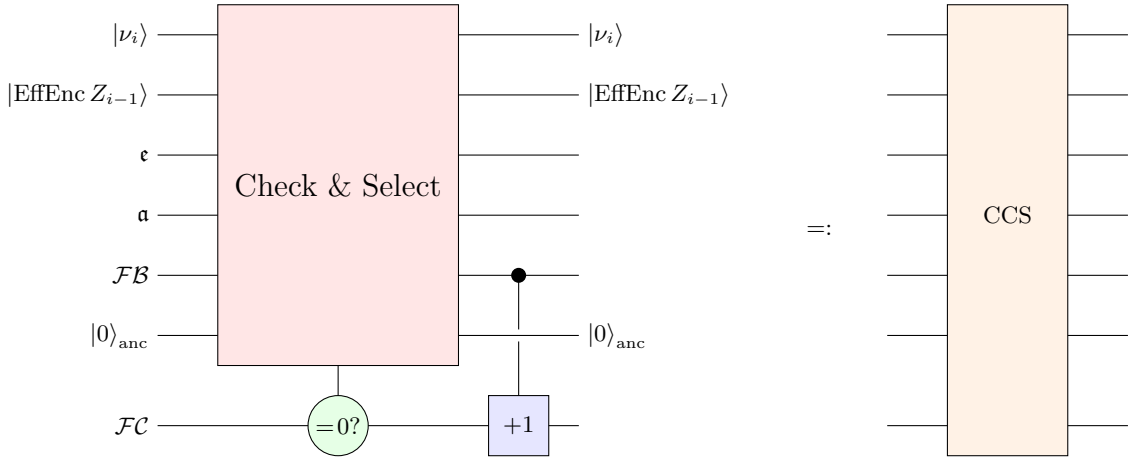

The basic building block of the implementation consists of a \emph{controlled check-and-select} (CCS) operation illustrated in Fig.~\ref{fig:checkselect}. We will implement seven different check-and-select operations, each corresponding to lines \ref{alg2case1-1}--\ref{alg2case1-2}, \ref{alg2case2-1}--\ref{alg2case2-2}, \ref{alg2case3-1}--\ref{alg2case3-2}, \ref{alg2case4-1}--\ref{alg2case4-2}, \ref{alg2case5-1}--\ref{alg2case5-2},   \ref{alg2case6-1}--\ref{alg2case6-2}, and \ref{alg2case7-1}--\ref{alg2case7-2}, of Alg.~\ref{alg:nonrecursiveEppstein}, respectively. For example, the check-and-select operation corresponding to lines \ref{alg2case1-1}--\ref{alg2case1-2} of Alg.~\ref{alg:nonrecursiveEppstein} does the following: if $G\backslash D$ contains a vertex with degree $2$ with at least one edge in $E\backslash (F\cup F'\cup D)$, it adds to $\mathfrak e$ the value of $e$ of one of the edges in $E\backslash (F\cup F'\cup D)$ incident to such a vertex (if multiple such vertex/edge combinations exist, $e$ is taken to be the first such edge of the first such vertex), leaves $\mathfrak a$ invariant (such that it stays in $0$), and flips $\mathcal{FB}$. 

\begin{figure}[htbp]
	\begin{tikzpicture}[scale=.8]

		\foreach \i in {-1,...,6}
		{
			\draw (0,\i) -- (2.7,\i);
			\draw (4.3,\i) --  (11.7,\i);
			\draw (13.3,\i) -- (16,\i);
			
			\node at (3.5,\i) {$\cdots$};
			\node at (12.5,\i) {$\cdots$};
		}		
		
		\foreach \i in {1,5,10,14}
		{
			\filldraw[fill=orange!10,draw] (\i -.2, 6.5) rectangle (\i+1.2,-.5);
		}
		
		\node at (1.5,3) {$\text{CCS}_1$};
		\node at (5.5,3) {$\text{CCS}_7$};
		
		\node at (10.5,3) {$\text{CCS}_7^{-1}$};
		\node at (14.5,3) {$\text{CCS}_1^{-1}$};

		\node[circle,inner sep=2pt,fill] at (7+.2,4) {};
		\node[circle,inner sep=2pt,fill] at (9-.2,3) {};
		
		\foreach \i in {0,...,3}
		{
			\filldraw[white] (7+.2-.1,\i-.1) rectangle (7+.2+.1,\i+.1);
		}
		
		\foreach \i in {0,...,2}
		{
			\filldraw[white] (9-.2-.1,\i-.1) rectangle (9-.2+.1,\i+.1);
		}
		
		\draw (7+.2,4) -- (7+.2,-1.3);
		\draw (9-.2,3) -- (9-.2,-1);
		
		\draw (7+.2,-1) circle (.3);
		
		\filldraw[fill=blue!10,draw] (8.5-.2,-.5) rectangle (9.5-.2,-1.5);
		\node at (9-.2,-1) {$+|E|$};
				
		\node[left] at (0,0) {$\ket0_{\mathcal{FC}}$};
		\node[right] at (16,0) {$\ket0_{\mathcal{FC}}$};
				
		\node[left] at (0,4) {$\ket0_{\mathfrak e}$};
		\node[right] at (16,4) {$\ket0_{\mathfrak e}$};
		
		\node[left] at (0,3) {$\ket0_{\mathfrak a}$};
		\node[right] at (16,3) {$\ket0_{\mathfrak a}$};
		
		\node[left] at (0,2) {$\ket{0}_{\mathcal{FB}}$};
		\node[right] at (16,2) {$\ket{0}_{\mathcal{FB}}$};
		
		\node[left] at (0,1) {$\ket{0}_{\text{anc}}$};
		\node[right] at (16,1) {$\ket{0}_{\text{anc}}$};
		
		\node[left] at (0,6) {$\ket{\nu_i}$};
		\node[right] at (16,6) {$\ket{\nu_i}$};

		\node[left] at (0,5) {$\ket{\EffEnc Z_{i-1}}$};
		\node[right] at (16,5) {$\ket{\EffEnc Z_{i-1}}$};
		
		\node[left] at (0,-1) {$\ket 0$};
		\node[right] at (16,-1) {$\ket{x_i}$};

	\end{tikzpicture}
	\caption{Reversible implementation of $\calculate_{G,F,i}$ using $O(\log n)$ ancillas. The seven CCS operations use the check-and-select operations corresponding to  lines \ref{alg2case1-1}--\ref{alg2case1-2}, \ref{alg2case2-1}--\ref{alg2case2-2}, \ref{alg2case3-1}--\ref{alg2case3-2}, \ref{alg2case4-1}--\ref{alg2case4-2}, \ref{alg2case5-1}--\ref{alg2case5-2},   \ref{alg2case6-1}--\ref{alg2case6-2}, and \ref{alg2case7-1}--\ref{alg2case7-2}, of Alg.~\ref{alg:nonrecursiveEppstein}, respectively.}\label{fig:calculatei}
\end{figure}
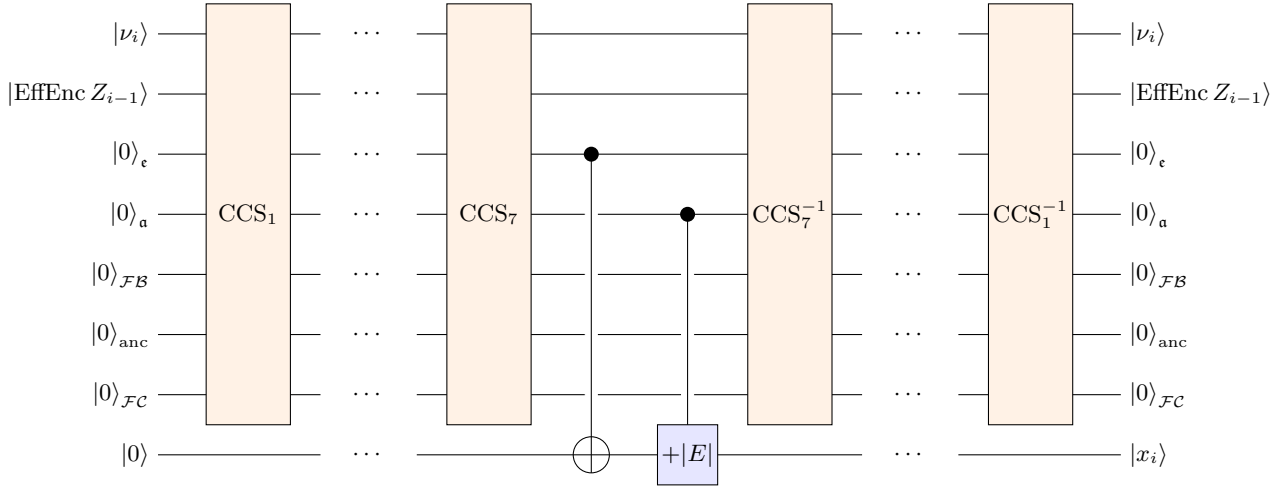

The full reversible implementation of $\calculate_{G,F,i}$ is given in Fig.~\ref{fig:calculatei}. It consists of these seven CCS operations concatenated in sequence, then adds $\mathfrak e + \mathfrak a |E|$ to the output register, and finally applies the inverses of the CCS operations. The latter ensure that all registers except for the input registers ($\ket{\nu_i}\ket{\EffEnc Z_{i-1}}$) and the ouput register are reset to $\ket 0$.

It thus remains to show how to reversibly implement each of the seven check-and-select operations. It is easy to see, however, that each of these can be implemented using $O(\log n)$ ancilla bits. 
Indeed, all check-and-select operations can be formulated as a search over a set of objects that can be classically enumerated beforehand, since $G$ and $F$ are classical inputs, and the checking of each individual object can always be implemented using a constant number of ancillas and calls to $\effcontains_{i-1}$. 

For example, checking whether any given vertex $v$ has any of the properties in question (e.g., whether $v$ has degree $2$ with at least one incident edge in $E\backslash (F\cup F'\cup D)$) reduces to checking which, if any, of its incident edges $e$ have been forced (i.e, $e\in X_{i-1}$) or deleted (i.e, $e+|E|\in X_{i-1}$). The latter can be done using a constant number of calls to $\effcontains_{i-1}$. Controlled on the outcome, a suitable edge and action bit is added to $\mathfrak e_{\text{temp}}$ and $\mathfrak a_{\text{temp}}$, respectively, and $\mathcal{FB}_{\text{temp}}$ flipped, where $\mathfrak e_{\text{temp}}$, $\mathfrak a_{\text{temp}}$ and $\mathcal{FB}_{\text{temp}}$ are ancilla registers of the same sizes as $\mathfrak e$, $\mathfrak a$, and $\mathcal{FB}$, respectively. 
The full CCS operation is then implemented by 
 introducing a counter from $0$ to $n$  initially set to $0$, and then, controlled on the counter being $0$, performing the previous operation for all $n$ vertices in sequence, followed by incrementing that counter controlled on $\mathcal{FB_{\text{temp}}}$ (this  is conceptually the same as the sequence of CCS operations in Fig.~\ref{fig:checkselect}--\ref{fig:calculatei}, except with $\mathcal{FC}$ replaced by a counter from $0$ to $n$, and $\mathfrak e$, $\mathfrak a$, and $\mathcal{FB}$ replaced by $\mathfrak e_{\text{temp}}$, $\mathfrak a_{\text{temp}}$, and $\mathcal{FB}_{\text{temp}}$, respectively). Note that the final value of the counter is then $n+1-$ the first vertex with the property (if the counter is $0$, no vertex has that property). 
Then, the values of $\mathfrak e_{\text{temp}}$, $\mathfrak a_{\text{temp}}$ and $\mathcal{FB}_{\text{temp}}$ are added to $\mathfrak e$, $\mathfrak a$, and $\mathcal{FB}$, respectively. 
After that, applying the inverse of the property checks for all vertices resets all ancillas. This covers the CCS operations corresponding to lines~\ref{alg2case1-1}--\ref{alg2case1-2}, \ref{alg2case2-1}--\ref{alg2case2-2},    \ref{alg2case6-1}--\ref{alg2case6-2},  and all but one of the subcases of  lines~\ref{alg2case4-1}--\ref{alg2case4-2} of Alg.~\ref{alg:nonrecursiveEppstein}.

A similar iteration through all edges in $E$ covers the CCS operation corresponding to lines~ \ref{alg2case7-1}--\ref{alg2case7-2} of Alg.~\ref{alg:nonrecursiveEppstein}.

One can similarly check if $G\backslash (F\cup F'\cup D)$ contains $4$-cycles with certain properties, and select incident edges accordingly if it does. Indeed, note that $G\backslash F$ has at most $O(n)$ $4$-cycles, which can be enumerated in a pre-specified order, so a counter using $O(\log n)$ bits suffices. This covers the CCS operations corresponding to lines \ref{alg2case3-1}--\ref{alg2case3-2} and  \ref{alg2case5-1}--\ref{alg2case5-2} of Alg.~\ref{alg:nonrecursiveEppstein}. 

Finally, to check if $G\backslash (F\cup F'\cup D)$ is a collection of disjoint $4$-cycles and isolated vertices, note that this is the case if and only if every edge in $E\backslash (F\cup F'\cup D)$ is part of exactly one $4$-cycle in $G\backslash (F\cup F'\cup D)$. This can thus be checked in a similar manner by sequentially going through all edges $e$, and then, for each (not necessarily unforced-isolated) $4$-cycle $\omega\ni e$ in $G\backslash F$, checking if any of its edges are in $F'\cup D$. This covers the final subcase of the CCS operation corresponding  lines~\ref{alg2case4-1}--\ref{alg2case4-2} of Alg.~\ref{alg:nonrecursiveEppstein}. 
\end{proof}

Using Theorem~\ref{thm:setgen}, we obtain the efficient reversible implementation of $\reduce$ from Proposition~\ref{prop:calculate}.

\begin{corollary}
	Let $(G,F)$ be a trivial-reduction-free FCHC instance. Then, the operation $\reduce_{G,F}$ defined by Eq.~\eqref{eq:reducecheckcyclesdefine} can be implemented reversibly with $O(s\log(n/s) + s + \log n)$ ancillas and $O(\poly(n))$ gates, where $s=s(G,F)$ and $n=n(G)$. 
\end{corollary}

\subsubsection{Reversible space-efficient implementation of $\checkcycles_{G,F}$}

The reversible and space-efficient implementation of $\checkcycles_{G,F}$ follows the (non-reversible) implementation of $\checkcycles$ in Alg.~\ref{alg:nonrecursiveEppstein} by checking in turn whether the conditions in lines~\ref{algcheckcase1}, \ref{algcheckcase2}, \ref{algcheckcase3}, and \ref{algcheckcase4}, respectively, apply.

It is clear that checking the first three conditions of $\checkcycles$, namely whether $G'$ contains a vertex of degree $1$, whether $F\cup F'$ contains three edges meeting at a vertex, or whether $G\backslash (F\cup F'\cup D)$ is a collection of disjoint $4$-cycles and isolated vertices, can each be done in the same way as  the check-and-select operations in the proof of Proposition~\ref{prop:calculate}, and thus can be implemented reversibly using $O(\log n(G))$ ancillas and $O(\poly(n(G)))$ gates. Hence, 
it only remains to check  whether $G\backslash D$ is connected. We use the fact that there is a classical reversible space-efficient algorithm to check that. Indeed, \cite{Reingold2005} shows that there is a (not necessarily reversible) classical algorithm that decides in time $O(\poly(n(G)))$ and space $O(\log n(G))$ if a graph $G$ is connected. Ref.~\cite{LMT2000} then shows that any given (not necessarily reversible) computation requiring memory $S$ and time $T$ can be implemented reversibly using $S$ ancillas and $2^{O(S)}$ gates (see also \cite{Williams2000,BTV2001}). This implies the following. 

\begin{proposition}[\cite{Reingold2005,LMT2000}]\label{prop:connectivity}
	There exists a classical deterministic algorithm that checks if a given graph $G$ is connected in time $O(\poly n(G))$, which can be implemented reversibly using $O(\log n(G))$ ancillas. 
\end{proposition}

Note that Proposition~\ref{prop:connectivity} only requires oracular access to the adjacency matrix of the graph in question \cite{Reingold2005}. Moreover, we can easily access the required adjacency matrix, i.e., the map
\begin{equation}\label{eq:adjacencymatrix}
	\ket{\EffEnc Z(\vec\nu)}\ket v\ket w\ket 0 \mapsto \ket{\EffEnc Z(\vec\nu)}\ket v\ket w\ket {v\sim_{G'} w}
\end{equation}
can be implemented reversibly with $O(1)$ ancillas and at most three call to $\effcontains_r$ , where the last bit in \eqref{eq:adjacencymatrix} is $1$ if the vertices $v$ and $w$ are connected in $G'=G\backslash D$, and $0$ otherwise. 
This completes our space-efficient reversible implementation of $\checkcycles_{G,F}$. 

\begin{corollary}
	Let $(G,F)$ be a trivial-reduction-free FCHC instance. Then, the operation $\checkcycles_{G,F}$ defined by Eq.~\eqref{eq:reducecheckcyclesdefine} can be implemented reversibly using $O(\log n(G))$ ancillas and $O(\poly(n(G)))$ gates. 
\end{corollary}

This concludes the proof of Theorem~\ref{thm:quantumAlg}.\qed

\section{Conclusion and outlook}\label{sec:conclusion}

The recent progress in experimental quantum computing \cite{2017_IBM,2018_Google,2018_Intel} increases confidence that  fully scalable quantum computers will be realised at some point in the upcoming decades. However, the  rate at which the number of qubits we can manipulate with relevant precision and coherence times currently grows provides significant motivation for studying potential uses of  size-limited  quantum computers.
Complementary to research dedicated to solving small-yet-hard simulation and ground-state problems \cite{Lloyd96,PhysRevA.92.042303,doi:10.1063/1.5027484}, which are promising applications for really small quantum computers, in this work, we  investigated ways to achieve speedups of classical algorithms by exploiting quantum computers significantly smaller than the problem size. 
 Concretely, we provide a framework for designing hybrid quantum-classical algorithms, which can allow for polynomial asymptotic speedups given a quantum computer which is any constant fraction of the problem size. 

Our result also implies that we can achieve a trade-off between the speedup we obtain, and the size of problem we wish to tackle. Thus, a small quantum computer can dramatically speed up the solving of small problems, but  can be used to achieve more modest speedups of larger instances as well. Such trade-offs have, to our knowledge, not been explored before our works.

We provided the general formalism in the form of the so-called divide-and-conquer hybrid approach,  which enables us to realise such trade-offs for a broad class of recursive classical algorithms, and we provided a characterisation of the space-efficiency of the quantum subroutines required to achieve polynomial speedups.
Moreover, we provided a toolkit for the space-efficient reversible generation and manipulation of sets, which is often the bottleneck of the space requirements of many such algorithms. As an illustration, we show how this framework can be applied to speed up the algorithm of Eppstein for detecting Hamilton cycles in cubic graphs.

We also identify a number of questions that remain unresolved, both from a purely theoretical and from an applied perspective.

First, the algorithms which we so far have applied our hybridd approach to are not the absolutely best known algorithms for their respective problems. Indeed, Eppstein's algorithm has subsequently been improved from $O(1.2599^n\poly(n))$ to $O(1.2509^n\poly(n))$ \cite{IwamaNakashima2007} and $O(1.2312^n\poly(n))$ \cite{XiaoNagamochi2016}. It would be interesting to  apply our framework to speed up the actually best classical algorithms, which would yield an asymptotic speedup over the best classical algorithms for a quantum computer the size of any constant fraction of the problem size.

Second, our approach currently focuses on Grover-based speedups; however, the quantum backtracking techniques  \cite{v014a015,PhysRevA.95.032323} and subsequent improvements \cite{Ambainis:2017:QAT:3055399.3055444} 
lead to better performing quantum algorithms. It remains an open question whether these methods can also be made to fit in our hybrid approach, and whether they would thus yield better speedups.

From a more applied perspective, there are two key issues which currently prevent our algorithms from being practical. First is the fact that we deal with asymptotic speedups, and focus on algorithm performance in the worst-case exponential run-times. For the issue of asymptotic statements,  there has lately been increasing interest in analysing performance for finite-size settings \cite{arXiv:1810.05582}, which were moderately promising, but required arbitrarily-sized quantum computers. It would be interesting to see if a similar claim could be made for size-limited quantum computers.

 Second, we assume ideal noiseless settings, which is still remote \cite{Preskill2018quantumcomputingin}. 
Note that exponential run-times essentially require full fault tolerance to yield reliable results. It would be interesting to consider our hybrid approach for heuristic algorithms, which run for low-polynomial times, and while they may fail to find solutions for most truly hard instances, still perform very well in practice. 
These much more efficient algorithms would be more important for real-world solutions to NP-hard problems. Furthermore, in this case, it is more likely that intermediary efficient error mitigation schemes, as opposed to full fault tolerance, suffices to achieve quantum-enhanced and usable NP heuristics.

\acknowledgments
We thank J.I.~Cirac for helpful discussions. 
VD is partially funded through the Quantum Software Consortium.


\begin{thebibliography}{29}%
\makeatletter
\providecommand \@ifxundefined [1]{%
 \@ifx{#1\undefined}
}%
\providecommand \@ifnum [1]{%
 \ifnum #1\expandafter \@firstoftwo
 \else \expandafter \@secondoftwo
 \fi
}%
\providecommand \@ifx [1]{%
 \ifx #1\expandafter \@firstoftwo
 \else \expandafter \@secondoftwo
 \fi
}%
\providecommand \natexlab [1]{#1}%
\providecommand \enquote  [1]{``#1''}%
\providecommand \bibnamefont  [1]{#1}%
\providecommand \bibfnamefont [1]{#1}%
\providecommand \citenamefont [1]{#1}%
\providecommand \href@noop [0]{\@secondoftwo}%
\providecommand \href [0]{\begingroup \@sanitize@url \@href}%
\providecommand \@href[1]{\@@startlink{#1}\@@href}%
\providecommand \@@href[1]{\endgroup#1\@@endlink}%
\providecommand \@sanitize@url [0]{\catcode `\\12\catcode `\$12\catcode
  `\&12\catcode `\#12\catcode `\^12\catcode `\_12\catcode `\%12\relax}%
\providecommand \@@startlink[1]{}%
\providecommand \@@endlink[0]{}%
\providecommand \url  [0]{\begingroup\@sanitize@url \@url }%
\providecommand \@url [1]{\endgroup\@href {#1}{\urlprefix }}%
\providecommand \urlprefix  [0]{URL }%
\providecommand \Eprint [0]{\href }%
\providecommand \doibase [0]{http://dx.doi.org/}%
\providecommand \selectlanguage [0]{\@gobble}%
\providecommand \bibinfo  [0]{\@secondoftwo}%
\providecommand \bibfield  [0]{\@secondoftwo}%
\providecommand \translation [1]{[#1]}%
\providecommand \BibitemOpen [0]{}%
\providecommand \bibitemStop [0]{}%
\providecommand \bibitemNoStop [0]{.\EOS\space}%
\providecommand \EOS [0]{\spacefactor3000\relax}%
\providecommand \BibitemShut  [1]{\csname bibitem#1\endcsname}%
\let\auto@bib@innerbib\@empty
\bibitem [{\citenamefont {Dunjko}\ \emph {et~al.}(2018)\citenamefont {Dunjko},
  \citenamefont {Ge},\ and\ \citenamefont {Cirac}}]{DGC2018}%
  \BibitemOpen
  \bibfield  {author} {\bibinfo {author} {\bibfnamefont {V.}~\bibnamefont
  {Dunjko}}, \bibinfo {author} {\bibfnamefont {Y.}~\bibnamefont {Ge}}, \ and\
  \bibinfo {author} {\bibfnamefont {J.~I.}\ \bibnamefont {Cirac}},\ }\href
  {\doibase 10.1103/PhysRevLett.121.250501} {\bibfield  {journal} {\bibinfo
  {journal} {Phys. Rev. Lett.}\ }\textbf {\bibinfo {volume} {121}},\ \bibinfo
  {pages} {250501} (\bibinfo {year} {2018})}\BibitemShut {NoStop}%
\bibitem [{\citenamefont {Yung}\ \emph {et~al.}(2014)\citenamefont {Yung},
  \citenamefont {Whitfield}, \citenamefont {Boixo}, \citenamefont {Tempel},\
  and\ \citenamefont {Aspuru-Guzik}}]{QChemIntro}%
  \BibitemOpen
  \bibfield  {author} {\bibinfo {author} {\bibfnamefont {M.-H.}\ \bibnamefont
  {Yung}}, \bibinfo {author} {\bibfnamefont {J.~D.}\ \bibnamefont {Whitfield}},
  \bibinfo {author} {\bibfnamefont {S.}~\bibnamefont {Boixo}}, \bibinfo
  {author} {\bibfnamefont {D.~G.}\ \bibnamefont {Tempel}}, \ and\ \bibinfo
  {author} {\bibfnamefont {A.}~\bibnamefont {Aspuru-Guzik}},\ }\enquote
  {\bibinfo {title} {Introduction to quantum algorithms for physics and
  chemistry},}\ in\ \href {\doibase 10.1002/9781118742631.ch03} {\emph
  {\bibinfo {booktitle} {Quantum Information and Computation for Chemistry}}}\
  (\bibinfo  {publisher} {John Wiley \& Sons, Inc.},\ \bibinfo {year} {2014})\
  pp.\ \bibinfo {pages} {67--106}\BibitemShut {NoStop}%
\bibitem [{\citenamefont {Bravyi}\ \emph {et~al.}(2016)\citenamefont {Bravyi},
  \citenamefont {Smith},\ and\ \citenamefont {Smolin}}]{2016_Bravyi}%
  \BibitemOpen
  \bibfield  {author} {\bibinfo {author} {\bibfnamefont {S.}~\bibnamefont
  {Bravyi}}, \bibinfo {author} {\bibfnamefont {G.}~\bibnamefont {Smith}}, \
  and\ \bibinfo {author} {\bibfnamefont {J.~A.}\ \bibnamefont {Smolin}},\
  }\href {\doibase 10.1103/PhysRevX.6.021043} {\bibfield  {journal} {\bibinfo
  {journal} {Phys. Rev. X}\ }\textbf {\bibinfo {volume} {6}},\ \bibinfo {pages}
  {021043} (\bibinfo {year} {2016})}\BibitemShut {NoStop}%
\bibitem [{\citenamefont {Peng}\ \emph {et~al.}(2019)\citenamefont {Peng},
  \citenamefont {Harrow}, \citenamefont {Ozols},\ and\ \citenamefont
  {Wu}}]{arXiv:1904.00102}%
  \BibitemOpen
  \bibfield  {author} {\bibinfo {author} {\bibfnamefont {T.}~\bibnamefont
  {Peng}}, \bibinfo {author} {\bibfnamefont {A.}~\bibnamefont {Harrow}},
  \bibinfo {author} {\bibfnamefont {M.}~\bibnamefont {Ozols}}, \ and\ \bibinfo
  {author} {\bibfnamefont {X.}~\bibnamefont {Wu}},\ }\href@noop {} {\enquote
  {\bibinfo {title} {Simulating large quantum circuits on a small quantum
  computer},}\ } (\bibinfo {year} {2019}),\ \Eprint
  {http://arxiv.org/abs/arXiv:1904.00102} {arXiv:1904.00102} \BibitemShut
  {NoStop}%
\bibitem [{\citenamefont {{Eppstein}}(2007)}]{Eppstein}%
  \BibitemOpen
  \bibfield  {author} {\bibinfo {author} {\bibfnamefont {D.}~\bibnamefont
  {{Eppstein}}},\ }\href {\doibase 10.7155/jgaa.00137} {\bibfield  {journal}
  {\bibinfo  {journal} {Journal of Graph Algorithms and Applications}\ }\textbf
  {\bibinfo {volume} {11}},\ \bibinfo {pages} {61} (\bibinfo {year}
  {2007})}\BibitemShut {NoStop}%
\bibitem [{\citenamefont {Moylett}\ \emph {et~al.}(2017)\citenamefont
  {Moylett}, \citenamefont {Linden},\ and\ \citenamefont
  {Montanaro}}]{PhysRevA.95.032323}%
  \BibitemOpen
  \bibfield  {author} {\bibinfo {author} {\bibfnamefont {D.~J.}\ \bibnamefont
  {Moylett}}, \bibinfo {author} {\bibfnamefont {N.}~\bibnamefont {Linden}}, \
  and\ \bibinfo {author} {\bibfnamefont {A.}~\bibnamefont {Montanaro}},\ }\href
  {\doibase 10.1103/PhysRevA.95.032323} {\bibfield  {journal} {\bibinfo
  {journal} {Phys. Rev. A}\ }\textbf {\bibinfo {volume} {95}},\ \bibinfo
  {pages} {032323} (\bibinfo {year} {2017})}\BibitemShut {NoStop}%
\bibitem [{\citenamefont {Bentley}\ \emph {et~al.}(1980)\citenamefont
  {Bentley}, \citenamefont {Haken},\ and\ \citenamefont
  {Saxe}}]{Bentley:1980:GMS:1008861.1008865}%
  \BibitemOpen
  \bibfield  {author} {\bibinfo {author} {\bibfnamefont {J.~L.}\ \bibnamefont
  {Bentley}}, \bibinfo {author} {\bibfnamefont {D.}~\bibnamefont {Haken}}, \
  and\ \bibinfo {author} {\bibfnamefont {J.~B.}\ \bibnamefont {Saxe}},\ }\href
  {\doibase 10.1145/1008861.1008865} {\bibfield  {journal} {\bibinfo  {journal}
  {SIGACT News}\ }\textbf {\bibinfo {volume} {12}},\ \bibinfo {pages} {36}
  (\bibinfo {year} {1980})}\BibitemShut {NoStop}%
\bibitem [{\citenamefont {Ambainis}(2004)}]{2004_Ambainis}%
  \BibitemOpen
  \bibfield  {author} {\bibinfo {author} {\bibfnamefont {A.}~\bibnamefont
  {Ambainis}},\ }\href {\doibase 10.1145/992287.992296} {\bibfield  {journal}
  {\bibinfo  {journal} {SIGACT News}\ }\textbf {\bibinfo {volume} {35}},\
  \bibinfo {pages} {22} (\bibinfo {year} {2004})}\BibitemShut {NoStop}%
\bibitem [{\citenamefont {Moser}\ and\ \citenamefont
  {Scheder}(2011)}]{2011_Moser}%
  \BibitemOpen
  \bibfield  {author} {\bibinfo {author} {\bibfnamefont {R.~A.}\ \bibnamefont
  {Moser}}\ and\ \bibinfo {author} {\bibfnamefont {D.}~\bibnamefont
  {Scheder}},\ }in\ \href {\doibase 10.1145/1993636.1993670} {\emph {\bibinfo
  {booktitle} {Proceedings of the Forty-third Annual ACM Symposium on Theory of
  Computing}}},\ \bibinfo {series and number} {STOC '11}\ (\bibinfo
  {publisher} {ACM},\ \bibinfo {address} {New York, NY, USA},\ \bibinfo {year}
  {2011})\ pp.\ \bibinfo {pages} {245--252}\BibitemShut {NoStop}%
\bibitem [{\citenamefont {Dantsin}\ \emph {et~al.}(2002)\citenamefont
  {Dantsin}, \citenamefont {Goerdt}, \citenamefont {Hirsch}, \citenamefont
  {Kannan}, \citenamefont {Kleinberg}, \citenamefont {Papadimitriou},
  \citenamefont {Raghavan},\ and\ \citenamefont {Sch\"{o}ning}}]{2002_Dantsin}%
  \BibitemOpen
  \bibfield  {author} {\bibinfo {author} {\bibfnamefont {E.}~\bibnamefont
  {Dantsin}}, \bibinfo {author} {\bibfnamefont {A.}~\bibnamefont {Goerdt}},
  \bibinfo {author} {\bibfnamefont {E.~A.}\ \bibnamefont {Hirsch}}, \bibinfo
  {author} {\bibfnamefont {R.}~\bibnamefont {Kannan}}, \bibinfo {author}
  {\bibfnamefont {J.}~\bibnamefont {Kleinberg}}, \bibinfo {author}
  {\bibfnamefont {C.}~\bibnamefont {Papadimitriou}}, \bibinfo {author}
  {\bibfnamefont {P.}~\bibnamefont {Raghavan}}, \ and\ \bibinfo {author}
  {\bibfnamefont {U.}~\bibnamefont {Sch\"{o}ning}},\ }\href {\doibase
  https://doi.org/10.1016/S0304-3975(01)00174-8} {\bibfield  {journal}
  {\bibinfo  {journal} {Theoretical Computer Science}\ }\textbf {\bibinfo
  {volume} {289}},\ \bibinfo {pages} {69 } (\bibinfo {year}
  {2002})}\BibitemShut {NoStop}%
\bibitem [{\citenamefont {Sch\"{o}ning}(1999)}]{1999_Schoning}%
  \BibitemOpen
  \bibfield  {author} {\bibinfo {author} {\bibfnamefont {T.}~\bibnamefont
  {Sch\"{o}ning}},\ }in\ \href {\doibase 10.1109/SFFCS.1999.814612} {\emph
  {\bibinfo {booktitle} {40th Annual Symposium on Foundations of Computer
  Science (Cat. No.99CB37039)}}}\ (\bibinfo {year} {1999})\ pp.\ \bibinfo
  {pages} {410--414}\BibitemShut {NoStop}%
\bibitem [{\citenamefont {Buhrman}\ \emph {et~al.}(2001)\citenamefont
  {Buhrman}, \citenamefont {Tromp},\ and\ \citenamefont
  {Vit{\'a}nyi}}]{BTV2001}%
  \BibitemOpen
  \bibfield  {author} {\bibinfo {author} {\bibfnamefont {H.}~\bibnamefont
  {Buhrman}}, \bibinfo {author} {\bibfnamefont {J.}~\bibnamefont {Tromp}}, \
  and\ \bibinfo {author} {\bibfnamefont {P.}~\bibnamefont {Vit{\'a}nyi}},\ }in\
  \href@noop {} {\emph {\bibinfo {booktitle} {Automata, Languages and
  Programming}}},\ \bibinfo {editor} {edited by\ \bibinfo {editor}
  {\bibfnamefont {F.}~\bibnamefont {Orejas}}, \bibinfo {editor} {\bibfnamefont
  {P.~G.}\ \bibnamefont {Spirakis}}, \ and\ \bibinfo {editor} {\bibfnamefont
  {J.}~\bibnamefont {van Leeuwen}}}\ (\bibinfo  {publisher} {Springer Berlin
  Heidelberg},\ \bibinfo {address} {Berlin, Heidelberg},\ \bibinfo {year}
  {2001})\ pp.\ \bibinfo {pages} {1017--1027}\BibitemShut {NoStop}%
\bibitem [{\citenamefont {Brassard}\ \emph {et~al.}(2002)\citenamefont
  {Brassard}, \citenamefont {H{\o}yer}, \citenamefont {Mosca},\ and\
  \citenamefont {Tapp}}]{2000_Brassard}%
  \BibitemOpen
  \bibfield  {author} {\bibinfo {author} {\bibfnamefont {G.}~\bibnamefont
  {Brassard}}, \bibinfo {author} {\bibfnamefont {P.}~\bibnamefont {H{\o}yer}},
  \bibinfo {author} {\bibfnamefont {M.}~\bibnamefont {Mosca}}, \ and\ \bibinfo
  {author} {\bibfnamefont {A.}~\bibnamefont {Tapp}},\ }in\ \href@noop {} {\emph
  {\bibinfo {booktitle} {Quantum Computation and Information}}},\ \bibinfo
  {series} {AMS Contemporary Mathematics Series}, Vol.\ \bibinfo {volume}
  {305}\ (\bibinfo  {publisher} {AMS},\ \bibinfo {year} {2002})\BibitemShut
  {NoStop}%
\bibitem [{\citenamefont {Yoder}\ \emph {et~al.}(2014)\citenamefont {Yoder},
  \citenamefont {Low},\ and\ \citenamefont {Chuang}}]{YoderLowChuang14}%
  \BibitemOpen
  \bibfield  {author} {\bibinfo {author} {\bibfnamefont {T.~J.}\ \bibnamefont
  {Yoder}}, \bibinfo {author} {\bibfnamefont {G.~H.}\ \bibnamefont {Low}}, \
  and\ \bibinfo {author} {\bibfnamefont {I.~L.}\ \bibnamefont {Chuang}},\
  }\href {\doibase 10.1103/PhysRevLett.113.210501} {\bibfield  {journal}
  {\bibinfo  {journal} {Phys. Rev. Lett.}\ }\textbf {\bibinfo {volume} {113}},\
  \bibinfo {pages} {210501} (\bibinfo {year} {2014})}\BibitemShut {NoStop}%
\bibitem [{\citenamefont {Reingold}(2005)}]{Reingold2005}%
  \BibitemOpen
  \bibfield  {author} {\bibinfo {author} {\bibfnamefont {O.}~\bibnamefont
  {Reingold}},\ }in\ \href {\doibase 10.1145/1060590.1060647} {\emph {\bibinfo
  {booktitle} {Proceedings of the Thirty-seventh Annual ACM Symposium on Theory
  of Computing}}},\ \bibinfo {series and number} {STOC '05}\ (\bibinfo
  {publisher} {ACM},\ \bibinfo {address} {New York, NY, USA},\ \bibinfo {year}
  {2005})\ pp.\ \bibinfo {pages} {376--385}\BibitemShut {NoStop}%
\bibitem [{\citenamefont {Lange}\ \emph {et~al.}(2000)\citenamefont {Lange},
  \citenamefont {McKenzie},\ and\ \citenamefont {Tapp}}]{LMT2000}%
  \BibitemOpen
  \bibfield  {author} {\bibinfo {author} {\bibfnamefont {K.-J.}\ \bibnamefont
  {Lange}}, \bibinfo {author} {\bibfnamefont {P.}~\bibnamefont {McKenzie}}, \
  and\ \bibinfo {author} {\bibfnamefont {A.}~\bibnamefont {Tapp}},\ }\href
  {\doibase https://doi.org/10.1006/jcss.1999.1672} {\bibfield  {journal}
  {\bibinfo  {journal} {Journal of Computer and System Sciences}\ }\textbf
  {\bibinfo {volume} {60}},\ \bibinfo {pages} {354 } (\bibinfo {year}
  {2000})}\BibitemShut {NoStop}%
\bibitem [{\citenamefont {Williams}(2000)}]{Williams2000}%
  \BibitemOpen
  \bibfield  {author} {\bibinfo {author} {\bibfnamefont {R.}~\bibnamefont
  {Williams}},\ }\href {https://people.csail.mit.edu/rrw/spacesim9_22.pdf}
  {\enquote {\bibinfo {title} {Space-efficient reversible simulations},}\ }
  (\bibinfo {year} {2000})\BibitemShut {NoStop}%
\bibitem [{201(2017)}]{2017_IBM}%
  \BibitemOpen
  \href@noop {} {\enquote {\bibinfo {title} {Ieee spectrum. ibm edges closer to
  quantum supremacy with 50-qubit processor},}\ }\bibinfo {howpublished}
  {\url{https://spectrum.ieee.org/tech-talk/computing/hardware/ibm-edges-closer-to-quantum-supremacy-with-50qubit-processor}}
  (\bibinfo {year} {2017})\BibitemShut {NoStop}%
\bibitem [{201(2018{\natexlab{a}})}]{2018_Google}%
  \BibitemOpen
  \href@noop {} {\enquote {\bibinfo {title} {{American Physical Society}
  meeting. engineering superconducting qubit arrays for quantum supremacy},}\
  }\bibinfo {howpublished}
  {\url{http://meetings.aps.org/Meeting/MAR18/Session/A33.1}} (\bibinfo {year}
  {2018}{\natexlab{a}})\BibitemShut {NoStop}%
\bibitem [{201(2018{\natexlab{b}})}]{2018_Intel}%
  \BibitemOpen
  \href@noop {} {\enquote {\bibinfo {title} {Intel newsroom. 2018 ces: Intel
  advances quantum and neuromorphic computing research},}\ }\bibinfo
  {howpublished}
  {\url{https://newsroom.intel.com/news/intel-advances-quantum-neuromorphic-computing-research/}}
  (\bibinfo {year} {2018}{\natexlab{b}})\BibitemShut {NoStop}%
\bibitem [{\citenamefont {Lloyd}(1996)}]{Lloyd96}%
  \BibitemOpen
  \bibfield  {author} {\bibinfo {author} {\bibfnamefont {S.}~\bibnamefont
  {Lloyd}},\ }\href {\doibase 10.1126/science.273.5278.1073} {\bibfield
  {journal} {\bibinfo  {journal} {Science}\ }\textbf {\bibinfo {volume}
  {273}},\ \bibinfo {pages} {1073} (\bibinfo {year} {1996})}\BibitemShut
  {NoStop}%
\bibitem [{\citenamefont {Wecker}\ \emph {et~al.}(2015)\citenamefont {Wecker},
  \citenamefont {Hastings},\ and\ \citenamefont {Troyer}}]{PhysRevA.92.042303}%
  \BibitemOpen
  \bibfield  {author} {\bibinfo {author} {\bibfnamefont {D.}~\bibnamefont
  {Wecker}}, \bibinfo {author} {\bibfnamefont {M.~B.}\ \bibnamefont
  {Hastings}}, \ and\ \bibinfo {author} {\bibfnamefont {M.}~\bibnamefont
  {Troyer}},\ }\href {\doibase 10.1103/PhysRevA.92.042303} {\bibfield
  {journal} {\bibinfo  {journal} {Phys. Rev. A}\ }\textbf {\bibinfo {volume}
  {92}},\ \bibinfo {pages} {042303} (\bibinfo {year} {2015})}\BibitemShut
  {NoStop}%
\bibitem [{\citenamefont {Ge}\ \emph {et~al.}(2019)\citenamefont {Ge},
  \citenamefont {Tura},\ and\ \citenamefont {Cirac}}]{doi:10.1063/1.5027484}%
  \BibitemOpen
  \bibfield  {author} {\bibinfo {author} {\bibfnamefont {Y.}~\bibnamefont
  {Ge}}, \bibinfo {author} {\bibfnamefont {J.}~\bibnamefont {Tura}}, \ and\
  \bibinfo {author} {\bibfnamefont {J.~I.}\ \bibnamefont {Cirac}},\ }\href
  {\doibase 10.1063/1.5027484} {\bibfield  {journal} {\bibinfo  {journal}
  {Journal of Mathematical Physics}\ }\textbf {\bibinfo {volume} {60}},\
  \bibinfo {pages} {022202} (\bibinfo {year} {2019})}\BibitemShut {NoStop}%
\bibitem [{\citenamefont {Iwama}\ and\ \citenamefont
  {Nakashima}(2007)}]{IwamaNakashima2007}%
  \BibitemOpen
  \bibfield  {author} {\bibinfo {author} {\bibfnamefont {K.}~\bibnamefont
  {Iwama}}\ and\ \bibinfo {author} {\bibfnamefont {T.}~\bibnamefont
  {Nakashima}},\ }in\ \href {http://dl.acm.org/citation.cfm?id=2394650.2394663}
  {\emph {\bibinfo {booktitle} {Proceedings of the 13th Annual International
  Conference on Computing and Combinatorics}}},\ \bibinfo {series and number}
  {COCOON'07}\ (\bibinfo  {publisher} {Springer-Verlag},\ \bibinfo {address}
  {Berlin, Heidelberg},\ \bibinfo {year} {2007})\ pp.\ \bibinfo {pages}
  {108--117}\BibitemShut {NoStop}%
\bibitem [{\citenamefont {Xiao}\ and\ \citenamefont
  {Nagamochi}(2016)}]{XiaoNagamochi2016}%
  \BibitemOpen
  \bibfield  {author} {\bibinfo {author} {\bibfnamefont {M.}~\bibnamefont
  {Xiao}}\ and\ \bibinfo {author} {\bibfnamefont {H.}~\bibnamefont
  {Nagamochi}},\ }\href {\doibase 10.1007/s00453-015-9970-4} {\bibfield
  {journal} {\bibinfo  {journal} {Algorithmica}\ }\textbf {\bibinfo {volume}
  {74}},\ \bibinfo {pages} {713} (\bibinfo {year} {2016})}\BibitemShut
  {NoStop}%
\bibitem [{\citenamefont {Montanaro}(2018)}]{v014a015}%
  \BibitemOpen
  \bibfield  {author} {\bibinfo {author} {\bibfnamefont {A.}~\bibnamefont
  {Montanaro}},\ }\href {\doibase 10.4086/toc.2018.v014a015} {\bibfield
  {journal} {\bibinfo  {journal} {Theory of Computing}\ }\textbf {\bibinfo
  {volume} {14}},\ \bibinfo {pages} {1} (\bibinfo {year} {2018})}\BibitemShut
  {NoStop}%
\bibitem [{\citenamefont {Ambainis}\ and\ \citenamefont
  {Kokainis}(2017)}]{Ambainis:2017:QAT:3055399.3055444}%
  \BibitemOpen
  \bibfield  {author} {\bibinfo {author} {\bibfnamefont {A.}~\bibnamefont
  {Ambainis}}\ and\ \bibinfo {author} {\bibfnamefont {M.}~\bibnamefont
  {Kokainis}},\ }in\ \href {\doibase 10.1145/3055399.3055444} {\emph {\bibinfo
  {booktitle} {Proceedings of the 49th Annual ACM SIGACT Symposium on Theory of
  Computing}}},\ \bibinfo {series and number} {STOC 2017}\ (\bibinfo
  {publisher} {ACM},\ \bibinfo {address} {New York, NY, USA},\ \bibinfo {year}
  {2017})\ pp.\ \bibinfo {pages} {989--1002}\BibitemShut {NoStop}%
\bibitem [{\citenamefont {Campbell}\ \emph {et~al.}(2018)\citenamefont
  {Campbell}, \citenamefont {Khurana},\ and\ \citenamefont
  {Montanaro}}]{arXiv:1810.05582}%
  \BibitemOpen
  \bibfield  {author} {\bibinfo {author} {\bibfnamefont {E.}~\bibnamefont
  {Campbell}}, \bibinfo {author} {\bibfnamefont {A.}~\bibnamefont {Khurana}}, \
  and\ \bibinfo {author} {\bibfnamefont {A.}~\bibnamefont {Montanaro}},\
  }\href@noop {} {\enquote {\bibinfo {title} {Applying quantum algorithms to
  constraint satisfaction problems},}\ } (\bibinfo {year} {2018}),\ \Eprint
  {http://arxiv.org/abs/arXiv:1810.05582} {arXiv:1810.05582} \BibitemShut
  {NoStop}%
\bibitem [{\citenamefont {Preskill}(2018)}]{Preskill2018quantumcomputingin}%
  \BibitemOpen
  \bibfield  {author} {\bibinfo {author} {\bibfnamefont {J.}~\bibnamefont
  {Preskill}},\ }\href {\doibase 10.22331/q-2018-08-06-79} {\bibfield
  {journal} {\bibinfo  {journal} {{Quantum}}\ }\textbf {\bibinfo {volume}
  {2}},\ \bibinfo {pages} {79} (\bibinfo {year} {2018})}\BibitemShut {NoStop}%
\end{thebibliography}

%

\appendix

\section{Runtime analysis of $\Eppstein$}\label{app:runtime}

In this section, we prove the runtime of $\Eppstein$ (as defined in Alg.~\ref{alg:eppsteinClassical}) of $O(2^{s(G,F)/3}\poly(n(G)))$. Note that since $s(G,F)\leq n(G)$, this in particular implies a runtime of $O(2^{n(G)/3}\poly(n(G)))$, which is the bound commonly cited in literature.  The analysis is essentially the same as that in \cite{Eppstein}. 

We first introduce a few notions for convenience.

\begin{definition}
	For any FCHC instance $(G,F)$, let $\trivred(G,F)$ be the FCHC instance obtained from $(G,F)$ by applying step~\ref{enum:classicalEppsteinStep1} of $\Eppstein$ to $(G,F)$. 
\end{definition}
In other words, $\trivred$ is the first step of $\Eppstein$ which performs all possible trivial reductions (step~\ref{enum:classicalEppsteinTrivialDeg2}, \ref{enum:classicalEppsteinTrivialDeg3}, and \ref{enum:classicalEppsteinTrivialCycle}) until no more such reductions are possible. In particular, $\trivred(G,F)$ is trivial-reduction-free for any FCHC instance $(G,F)$, and $\trivred(G,F)=(G,F)$ whenever $(G,F)$ is trivial-reduction-free.

\begin{definition}
	Let $(G,F)$ be a trivial-reduction-free FCHC instance. We say that $(G,F)$ is \emph{non-terminal} if 
	\begin{enumerate}[(i)]
		\item $G$ does not contain any vertices of degree $0$ or $1$,
		\item $F$ does not contain three edges meeting at a vertex, 
		\item $G\backslash F$ is not a collection of disjoint $4$-cycles and isolated vertices, and
		\item $F$ does not contain a non-Hamiltonian cycle. 
	\end{enumerate}
	Otherwise, we call $(G,F)$ \emph{terminal}. 
\end{definition}
In other words, a trivial-reduction-free FCHC instance is non-terminal if and only if none of the terminal conditions in step~\ref{enum:classicalEppsteinStep2} of $\Eppstein$ apply.

\begin{definition}
	For any trivial-reduction-free and non-terminal FCHC instance $(G,F)$, let $\edgeselect(G,F)$ be the edge selected in step~\ref{enum:classicalEppsteinStep3} of $\Eppstein(G,F)$. 
\end{definition}
 In other words, if $(G,F)$ is  trivial-reduction-free and non-terminal, then steps~\ref{enum:classicalEppsteinStep4} and \ref{enum:classicalEppsteinStep5} of $\Eppstein(G,F)$ call $\Eppstein(G, F\cup\{\edgeselect(G,F)\})$ and $\Eppstein(G\backslash \{\edgeselect(G,F)\},F)$, respectively. 
 
Eppstein's key idea to bounding the runtime was to show that with each recursive call, $s(G,F)$ reduces by a constant larger than $1$, leading to a runtime that is polynomially better than a trivial path-search algorithm (see also Lemma~7 in \cite{Eppstein}).

\begin{proposition} \label{prop:sbound}
	Let $(G,F)$ be a  trivial-reduction-free and non-terminal FCHC instance and let $e=\edgeselect(G,F)$. Suppose that $F$ is nonempty, and let $(G_1,F_1)= \trivred(G, F\cup \{e\})$ and $(G_2,F_2) = \trivred(G\backslash\{e\}, F)$. Suppose that $(G_1,F_1)$ and $(G_2,F_2)$ are non-terminal. Then,
	\begin{itemize}
		\item $s(G_1,F_1), s(G_2,F_2)\leq s(G,F)-3$, or
		\item $s(G_1,F_1)\leq s(G,F)-2 $ and $ s(G_2,F_2) \leq s(G,F)-5$, or vice-versa. 
	\end{itemize}
\end{proposition}
\begin{proof}
 We prove the claim by going through all possible cases that can occur in $\edgeselect(G,F)$. 
 
	\begin{figure}[htbp]
		\begin{tikzpicture}
		\node (Q) at (0,0) {\begin{tikzpicture}[scale=.5]
			\node[blackdot] (A) at (0,0) {};		
			\node[blackdot] (B) at (0,1) {};
			\node[blackdot] (C) at (1,1) {};		
			\node[blackdot] (D) at (1,0) {};		
		
			\node (a) at (-.6,-.6) {};
			\node (b) at (-.6,1.6) {};
			\node (c) at (1.6,1.6) {};
			\node (d) at (1.6,-.6) {};
			
			\node[below] at (c) {$e$};
			
			\draw (A) -- (B) -- (C) -- (D)--(A);
			\draw (C) -- (c)  ;
			\draw[very thick,red] (A) -- (a);
			\draw[very thick,red] (B) -- (b);
			\draw (D) -- (d);

		\end{tikzpicture} };
		
		\node (force) at (-1.5,-1.5) {\begin{tikzpicture}[scale=.5]
			\node[blackdot] (A) at (0,0) {};		
			\node[blackdot] (B) at (0,1) {};
			\node[blackdot] (C) at (1,1) {};		
			\node[blackdot] (D) at (1,0) {};		
		
			\node (a) at (-.6,-.6) {};
			\node (b) at (-.6,1.6) {};
			\node (c) at (1.6,1.6) {};
			\node (d) at (1.6,-.6) {};
			
			\draw (A) -- (B) -- (C) -- (D)--(A);
			\draw[very thick,red] (C) -- (c)  ;
			\draw[very thick,red] (A) -- (a);
			\draw[very thick,red] (B) -- (b);
			\draw[very thick,red] (D) -- (d);

		\end{tikzpicture} };

		\node (delete) at (1.5,-1.5) {\begin{tikzpicture}[scale=.5]
			\node[blackdot] (A) at (0,0) {};		
			\node[blackdot] (B) at (0,1) {};
			\node[blackdot] (C) at (1,1) {};		
			\node[blackdot] (D) at (1,0) {};		
		
			\node (a) at (-.6,-.6) {};
			\node (b) at (-.6,1.6) {};
			\node (c) at (1.6,1.6) {};
			\node (d) at (1.6,-.6) {};
			
			\draw[very thick,red] (B) -- (C) -- (D)--(A);
			\draw[dotted] (C) -- (c)  ;
			\draw[very thick,red] (A) -- (a);
			\draw[very thick,red] (B) -- (b);
			\draw[dotted] (D) -- (d);
			\draw[dotted] (A)--(B);

		\end{tikzpicture}};	
		
		\draw[blue,->] (-.65,-.65) -- ++(-.2,-.2);
		\draw[blue,->] (.65,-.65) -- ++(.2,-.2);
		
		\end{tikzpicture}
		\caption{If $\edgeselect(G,F)$ selects $e$ according to step~\ref{enum:classicalEppsteinStep3a} of $\Eppstein$, $C(G_1,F_1) \geq C(G,F)+1$, $|F_1|\geq |F|+2$, and $|F_2|\geq |F|+3$. }\label{fig:app3a}
	\end{figure}
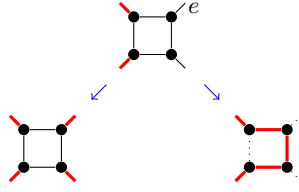

 Suppose first that $e$ is selected according to step~\ref{enum:classicalEppsteinStep3a} of $\Eppstein$ (see Fig.~\ref{fig:app3a}), i.e., $G\backslash F$ contains a $4$-cycle $\omega$ with exactly two vertices incident to edges in $F$, and $e$ is one of the other two edges adjacent to vertices in the cycle. Note that indeed, since $(G, F)$ is trivial-reduction-free, every vertex in $\omega$ is adjacent to a non-cycle edge. Note moreover that since $(G,F)$ is trivial-reduction-free, the two vertices of $\omega$ which are incident to edges in $F$ must be adjacent.  Hence, $\omega$ has two opposite vertices which are incident to an edge in $F\cup\{e\}$. Thus, $\trivred(G,F\cup\{e\})$ forces the last non-cycle edge adjacent to $\omega$. It follows that $|F_1|\geq |F|+2$ and $|C(G_1,F_1) |\geq |C(G,F)|+1$, and hence $s(G_1,F_1)\leq s(G,F)-3$. On the other hand, removing $e$ leads to two unforced edges in $\omega$, which will be forced, leading to one of the vertices in $\omega$ to have two incident forced edges. This leads to the removal of the third (cycle) edge, hence forcing the final edge in $\omega$. It follows that $|F_2|\geq |F|+3$, and hence $s(G_2,F_2)\leq s(G,F)-3$.
 
 Next, suppose that $e=yz$ is selected according to step~\ref{enum:classicalEppsteinStep3b} of $\Eppstein$.
	Let $w$ be the third vertex adjacent to $y$ in $G$. Note that since $(G,F)$ is trivial-reduction-free, $yw\not\in F$ and $\deg z=\deg w = 3$.  We distinguish two cases.

	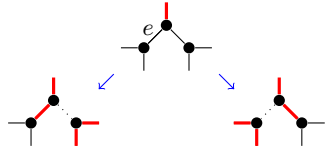
\begin{figure}[htbp]
		\begin{tikzpicture} 
		\node (Q) at (0,0) {\begin{tikzpicture}[scale=.3]
	
			\node[blackdot] (A) at (0,-1) {};		
			\node[blackdot] (B) at (-1,-2) {};
			\node[blackdot] (C) at (1,-2) {};		
			
			\node (T) at (0,.4) {};
			\node (L) at (-2.4,-2) {};
			\node (R) at (2.4,-2) {};
			\node (LB) at (-1,-3.4) {};
			\node (RB) at (1,-3.4) {};
			
			\draw[very thick,red] (A) -- (T);
	
			\path[draw] (A) edge node[above left=-4pt] {$e$} (B);
			
			\draw (A) -- (C);		
			
			\draw (B) -- (L);
			\draw (B) -- (LB);
			
			\draw (C) -- (R);
			\draw (C) -- (RB);
			
			\draw (A) -- (B)  ;
		\end{tikzpicture} };
		
		\node (force) at (-1.5,-1) {\begin{tikzpicture}[scale=.3]
	
			\node[blackdot] (A) at (0,-1) {};		
			\node[blackdot] (B) at (-1,-2) {};
			\node[blackdot] (C) at (1,-2) {};		
			
			\node (T) at (0,.4) {};
			\node (L) at (-2.4,-2) {};
			\node (R) at (2.4,-2) {};
			\node (LB) at (-1,-3.4) {};
			\node (RB) at (1,-3.4) {};
			
			\draw[very thick,red] (A) -- (T);
	
			\draw (A) -- (B);	
		
			\draw (B) -- (L);
			\draw (B) -- (LB);
		
			\draw (C) -- (R);
			\draw (C) -- (RB);
			
			\draw[very thick, red] (A) -- (B) ;
			
			\draw[dotted] (A) -- (C);
			
			\draw[very thick, red] (C) -- (R);
			\draw[very thick, red] (C) -- (RB);
		\end{tikzpicture} };

		\node (delete) at (1.5,-1) {\begin{tikzpicture}[scale=.3]
	
			\node[blackdot] (A) at (0,-1) {};		
			\node[blackdot] (B) at (-1,-2) {};
			\node[blackdot] (C) at (1,-2) {};		
			
			\node (T) at (0,.4) {};
			\node (L) at (-2.4,-2) {};
			\node (R) at (2.4,-2) {};
			\node (LB) at (-1,-3.4) {};
			\node (RB) at (1,-3.4) {};
			
			\draw[very thick,red] (A) -- (T);
	
			\draw (A) -- (C);		
			
			\draw (B) -- (L);
			\draw (B) -- (LB);
		
			\draw (C) -- (R);
			\draw (C) -- (RB);

			\draw[very thick, red] (A) -- (C) ;
			
			\draw[dotted] (A) -- (B);
			
			\draw[very thick, red] (B) -- (L);
			\draw[very thick, red] (B) -- (LB);
			
		\end{tikzpicture} };	
		
		\draw[blue,->] (-.7,-.5) -- ++(-.2,-.2);
		\draw[blue,->] (.7,-.5) -- ++(.2,-.2);

		\end{tikzpicture}
		\caption{If $\edgeselect(G,F)$ selects $e$ according to step~\ref{enum:classicalEppsteinStep3b} of $\Eppstein$, and neither $z$ nor $w$ have an incident edge in $F$, then $|F_1|,|F_2|\geq |F|+3$.} \label{fig:app3b-1}
	\end{figure}

	In the first case, suppose that neither $z$ nor $w$ have an incident edge in $F$ (Fig.~\ref{fig:app3b-1}). Then, $\trivred(G,F\cup\{e\})$ first removes $yw$ and then adds both remaining edges incident to $w$ to $F$. Hence, $s(G_1,F_1)\leq s(G,F)-3$. Similarly, $s(G_2,F_2)\leq s(G,F)-3$.
	
	In the second case, suppose that $z$ or $w$, have an incident edge in $F$. Note that since $xy\in F$, $yz$ and $yw$ cannot be part of a $4$-cycle of unforced edges, because otherwise $\edgeselect(G,F)$ would choose $e$ according to step~\ref{enum:classicalEppsteinStep3a} of $\Eppstein$. Thus, $z,y,w$ are part of an unforced, and possibly closed, chain of $k\geq 4$ vertices, with the inner vertices each having an incident edge in $F$. There are two subcases here. 

	\begin{figure}[htbp]
		\begin{subfigure}[b]{.3\textwidth}
		\begin{tikzpicture}
		\node (Q) at (0,0) {\begin{tikzpicture}[scale=.5]
	
			\node[blackdot] (A) at (0,0) {};		
			\node[blackdot] (B) at (1,0) {};
			\node[blackdot] (C) at (2,0) {};		
			\node[blackdot] (D) at (3,0) {};	

			\node (b) at (1,1) {};
			\node (c) at (2,1) {};
			
			\node (a1) at (-.7,-.7) {};
			\node (a2) at (-.7,.7) {};
			\node (d1) at (3.7,-.7) {};
			\node (d2) at (3.7,.7) {};
			
			\draw (A) -- (B) -- (C) -- (D); 
			
			\draw [very thick, red] (b) -- (B);
			\draw [very thick, red] (c) -- (C);
			
			\draw (a1) -- (A);
			\draw (a2) -- (A);
			\draw (d1) -- (D);
			\draw (d2) -- (D);
		\end{tikzpicture} };
		
		\node (force) at (-1.5,-1.5) {\begin{tikzpicture}[scale=.5]
	
			\node[blackdot] (A) at (0,0) {};		
			\node[blackdot] (B) at (1,0) {};
			\node[blackdot] (C) at (2,0) {};		
			\node[blackdot] (D) at (3,0) {};	

			\node (b) at (1,1) {};
			\node (c) at (2,1) {};
			
			\node (a1) at (-.7,-.7) {};
			\node (a2) at (-.7,.7) {};
			\node (d1) at (3.7,-.7) {};
			\node (d2) at (3.7,.7) {};
			
			\draw [very thick, red] (B) -- (C); 
			
			\draw[dotted] (A) -- (B);
			\draw[dotted] (C) -- (D);
			
			\draw [very thick, red] (b) -- (B);
			\draw [very thick, red] (c) -- (C);
			
			\draw[very thick,red] (a1) -- (A);
			\draw[very thick,red] (a2) -- (A);
			\draw[very thick,red] (d1) -- (D);
			\draw[very thick,red] (d2) -- (D);
		\end{tikzpicture} };

		\node (delete) at (1.5,-1.5) {\begin{tikzpicture}[scale=.5]
	
			\node[blackdot] (A) at (0,0) {};		
			\node[blackdot] (B) at (1,0) {};
			\node[blackdot] (C) at (2,0) {};		
			\node[blackdot] (D) at (3,0) {};	

			\node (b) at (1,1) {};
			\node (c) at (2,1) {};
			
			\node (a1) at (-.7,-.7) {};
			\node (a2) at (-.7,.7) {};
			\node (d1) at (3.7,-.7) {};
			\node (d2) at (3.7,.7) {};
			
			\draw[dotted]  (B) -- (C); 
			
			\draw[very thick, red] (A) -- (B);
			\draw[very thick, red] (C) -- (D);
			
			\draw [very thick, red] (b) -- (B);
			\draw [very thick, red] (c) -- (C);
			
			\draw (a1) -- (A);
			\draw (a2) -- (A);
			\draw (d1) -- (D);
			\draw (d2) -- (D);
		\end{tikzpicture} };	
		
		\draw[blue,->] (-1,-.7) -- ++(-.2,-.2);
		\draw[blue,->] (1,-.7) -- ++(.2,-.2);

		\end{tikzpicture}
			\caption{$k$ even}\label{subfig:keven}
		\end{subfigure}
		\qquad\qquad 
		\begin{subfigure}[b]{.3\textwidth}
		\begin{tikzpicture}
		\node (Q) at (0,0) {\begin{tikzpicture}[scale=.5]
	
			\node[blackdot] (A) at (0,0) {};		
			\node[blackdot] (B) at (1,0) {};
			\node[blackdot] (C) at (2,0) {};		
			\node[blackdot] (D) at (3,0) {};	
			\node[blackdot] (E) at (4,0) {};	

			\node (b) at (1,1) {};
			\node (c) at (2,1) {};
			\node (d) at (3,1) {};
					
			\node (a1) at (-.7,-.7) {};
			\node (a2) at (-.7,.7) {};
			\node (e1) at (4.7,-.7) {};
			\node (e2) at (4.7,.7) {};
			
			\draw (A) -- (B) -- (C) -- (D) -- (E); 
			
			\draw [very thick, red] (b) -- (B);
			\draw [very thick, red] (c) -- (C);
			\draw [very thick, red] (d) -- (D);
			
			\draw (a1) -- (A);
			\draw (a2) -- (A);
			\draw (e1) -- (E);
			\draw (e2) -- (E);
		\end{tikzpicture} };
		
		\node (force) at (-2,-1.5) {\begin{tikzpicture}[scale=.5]
	
			\node[blackdot] (A) at (0,0) {};		
			\node[blackdot] (B) at (1,0) {};
			\node[blackdot] (C) at (2,0) {};		
			\node[blackdot] (D) at (3,0) {};	
			\node[blackdot] (E) at (4,0) {};	

			\node (b) at (1,1) {};
			\node (c) at (2,1) {};
			\node (d) at (3,1) {};
					
			\node (a1) at (-.7,-.7) {};
			\node (a2) at (-.7,.7) {};
			\node (e1) at (4.7,-.7) {};
			\node (e2) at (4.7,.7) {};
			
			\draw [very thick, red] (B) -- (C); 
			\draw [very thick, red] (D) -- (E); 			
			
			\draw[dotted] (A) -- (B);
			\draw[dotted] (C) -- (D);
			
			\draw [very thick, red] (b) -- (B);
			\draw [very thick, red] (c) -- (C);
			\draw [very thick, red] (d) -- (D);
			
			\draw[very thick,red] (a1) -- (A);
			\draw[very thick,red] (a2) -- (A);
			\draw (e1) -- (E);
			\draw (e2) -- (E);
		\end{tikzpicture} };

		\node (delete) at (2,-1.5) {\begin{tikzpicture}[scale=.5]
	
			\node[blackdot] (A) at (0,0) {};		
			\node[blackdot] (B) at (1,0) {};
			\node[blackdot] (C) at (2,0) {};		
			\node[blackdot] (D) at (3,0) {};	
			\node[blackdot] (E) at (4,0) {};	

			\node (b) at (1,1) {};
			\node (c) at (2,1) {};
			\node (d) at (3,1) {};
					
			\node (a1) at (-.7,-.7) {};
			\node (a2) at (-.7,.7) {};
			\node (e1) at (4.7,-.7) {};
			\node (e2) at (4.7,.7) {};
			
			\draw [very thick, red] (A) -- (B); 
			\draw [very thick, red] (C) -- (D); 			
			
			\draw[dotted] (B) -- (C);
			\draw[dotted] (D) -- (E);
			
			\draw [very thick, red] (b) -- (B);
			\draw [very thick, red] (c) -- (C);
			\draw [very thick, red] (d) -- (D);
			
			\draw[very thick,red] (e1) -- (E);
			\draw[very thick,red] (e2) -- (E);
			\draw (a1) -- (A);
			\draw (a2) -- (A);
		\end{tikzpicture} };	
		
		\draw[blue,->] (-1,-.7) -- ++(-.2,-.2);
		\draw[blue,->] (1,-.7) -- ++(.2,-.2);

		\end{tikzpicture}
			\caption{$k$ odd}\label{subfig:kodd}
		\end{subfigure}
		\caption{$\edgeselect(G,F)$ selects $e$ according to step~\ref{enum:classicalEppsteinStep3b} of $\Eppstein$, and $z,y,w$ are part of an unforced chain of $k\geq 4$ vertices such that the inner $k-2$ vertices each have an incident edge in $F$ and the outer two vertices each have three unforced incident edges. (a) If $k$ is even, then $|F_1|\geq |F|+5$ and $|F_2|\geq |F|+2$ or vice-versa. (b) If $k$ is odd, then $|F_1|,|F_2| \geq |F|+4$.}\label{fig:app3b-2}
	\end{figure}
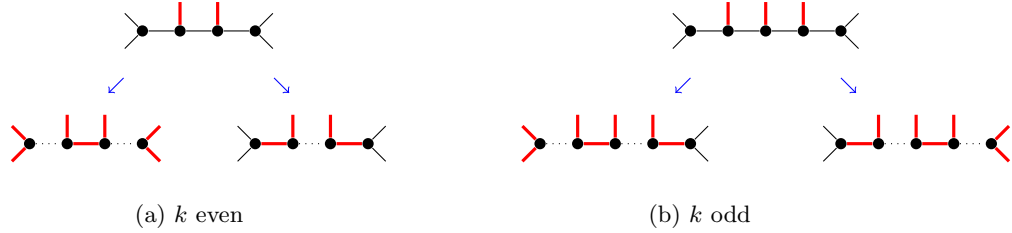

	In the first subcase, suppose that this chain terminates (see Fig.~\ref{fig:app3b-2}). Thus, $z,y,w$ are part of an unforced chain of $k\geq 4$ vertices such that the inner $k-2$ vertices each have an incident edge in $F$ and the outer two vertices have three unforced incident edges. Then, $\trivred(G, F\cup\{e\})$ and $\trivred(G\backslash\{e\},F)$ force and delete alternating edges of that chain.
	If $k=2l$ is even, then  one of $(G_1,F_1)$ and $(G_2,F_2)$ has $l-1$ edges of that chain which are forced, and the other has $l$ (Fig.~\ref{subfig:keven}). Assume without loss of generality that $(G_1,F_1)$ has $l-1$ edges of that chain forced. Then, both outer vertices of the chain will eventually have degree two, and thus $\trivred$ forces four additional edges. It follows that $|F_1| \geq |F|+4+(l-1)\geq |F|+5$ and $|F_2|\geq |F|+l\geq |F|+2$, and hence $s(G_1,F_1)\leq s(G,F)-5$ and $s(G_2,F_2)\leq s(G,F)-2$. 	
	On the other hand, if $k=2l+1$ is odd, then $\trivred(G, F\cup\{e\})$ and $\trivred(G\backslash\{e\},F)$ both force $l$ edges of the chain, and leave one of the two outer vertices with degree two, thus forcing its other two incident edges (Fig.~\ref{subfig:kodd}). It follows that $|F_1|,|F_2|\geq |F|+l+2\geq |F|+4$ and hence $s(G_1,F_1),s(G_2,F_2)\leq s(G,F)-4$. 	
	\begin{figure}[htbp]
		\begin{tikzpicture}
		\node (Q) at (0,0) {\begin{tikzpicture}[scale=.5]
	
			\foreach \i in {0,1,2,3,4,5}
			{
				\node[blackdot] (A\i) at (60*\i:1) {};
				\node (a\i) at (60*\i:1.8) {};
				
				\draw[very thick, red] (a\i) -- (A\i);
			}
			
			\draw (A0) -- (A1) -- (A2)--(A3)--(A4)--(A5)--(A0);
		\end{tikzpicture} };
		
		\node (force) at (-1.5,-1.5) {\begin{tikzpicture}[scale=.5]
			\foreach \i in {0,1,2,3,4,5}
			{
				\node[blackdot] (A\i) at (60*\i:1) {};
				\node (a\i) at (60*\i:1.8) {};
				
				\draw[very thick, red] (a\i) -- (A\i);
			}
			
			\draw[very thick,red] (A0) -- (A1);
			\draw[very thick,red] (A2) -- (A3);
			\draw[very thick,red] (A4) -- (A5);
			
			\draw[dotted] (A1) -- (A2);
			\draw[dotted] (A3) -- (A4);
			\draw[dotted] (A5) -- (A0);
		\end{tikzpicture} };

		\node (delete) at (1.5,-1.5) {\begin{tikzpicture}[scale=.5]
			\foreach \i in {0,1,2,3,4,5}
			{
				\node[blackdot] (A\i) at (60*\i:1) {};
				\node (a\i) at (60*\i:1.8) {};
				
				\draw[very thick, red] (a\i) -- (A\i);
			}
			
			\draw[dotted] (A0) -- (A1);
			\draw[dotted] (A2) -- (A3);
			\draw[dotted] (A4) -- (A5);
			
			\draw[very thick,red] (A1) -- (A2);
			\draw[very thick,red] (A3) -- (A4);
			\draw[very thick,red] (A5) -- (A0);
		\end{tikzpicture} };	
		
		\draw[blue,->] (-.7,-.7) -- ++(-.2,-.2);
		\draw[blue,->] (.7,-.7) -- ++(.2,-.2);

		\end{tikzpicture}
		\caption{If $\edgeselect(G,F)$ selects $e$ according to step~\ref{enum:classicalEppsteinStep3b} of $\Eppstein$, and $z,y,w$ are part of an unforced cycle of $k\geq 6$ vertices with $k$ even, each of which is incident to an edge in $F$, then $|F_1|,|F_2|\geq |F|+3$.} \label{fig:app3b-3}
	\end{figure}
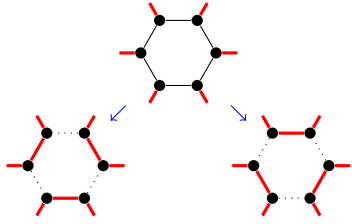
	
	In the second subcase, this unforced chain is a cycle of $k$ vertices (Fig.~\ref{fig:app3b-3}), each of which is incident to an edge in $F$. Then, $\trivred(G, F\cup\{e\})$ and $\trivred(G\backslash\{e\},F)$ force and delete alternating edges of that chain. Clearly, by the definition of Step~\ref{enum:classicalEppsteinStep3b} of $\Eppstein$, $k> 4$.  
	Moreover, if $k$ is odd, $\trivred(G,F\cup\{e\})$ and $\trivred(G\backslash\{e\},F)$ would both have three forced edges meeting in a vertex, and hence would be terminal. Hence, $k\geq 6$, and thus $|F_1|,|F_2|\geq |F|+3$. Hence, $s(G_1,F_1),s(G_2,F_2)\leq s(G,F)-3$.
 
 Finally, note that since $F$ is assumed to be non-empty and $(G,F)$ to be non-terminal, $\edgeselect(G,F)$ does not choose $e$  according to step~\ref{enum:classicalEppsteinStep3c} of $\Eppstein$. 
\end{proof}

\begin{corollary}
	For any FCHC instance $(G,F)$, $\Eppstein(G,F)$ decides the FCHC problem in a runtime of $O(2^{s(G,F)/3}\poly(n(G)))$. 
\end{corollary}
\begin{proof}
	Clearly, if $(G,F)$ is trivial-reduction-free and terminal, $\Eppstein(G,F)$ only takes $O(\poly(n(G)))$ time. Moreover, if $(G,F)$ is trivial-reduction-free, then by Propostion~\ref{prop:snonnegative}, $s(G,F) = 0$ implies that $(G,F)$ is terminal. Hence, Proposition~\ref{prop:sbound} implies that the runtime of $\Eppstein(G,F)$ can be bounded by some function $T(s(G,F))$ depending only on $s(G,F)$, where $T(s)$ satisfies $T(s)=O(\poly(n(G)))$ for $s\leq 0$, and 
	\begin{equation}
		T(s) \leq \max(2T(s-3), T(s-2)+T(s-5))
	\end{equation}
	for $s>0$. Using standard techniques for solving linear recurrence relations, one obtains $T(s)=O(2^{s/3}\poly(n(G)))$ for $s\geq 0$.
\end{proof}

\end{document}